\documentclass[a4paper]{article}

\usepackage{amsfonts}
\usepackage{amssymb}
\usepackage{graphicx}
\usepackage{amsmath}
\usepackage[a4paper]{geometry}
\usepackage{epstopdf}
\newtheorem{theorem}{Theorem}

\newtheorem{example}{Example}

\newtheorem{lemma}{Lemma}

\newtheorem{proposition}{Proposition}
\newtheorem{assumptions}{Assumptions}
\newtheorem{remark}{Remark}

\newenvironment{proof}[1][Proof]{\noindent\textbf{#1.} }{\ \rule{0.5em}{0.5em}}
\def\lessim{\ \lower4pt\hbox{$\buildrel{\displaystyle <}\over\sim$}\ }

\begin{document}
\begin{center}
{\LARGE Simultaneous sparse model selection and coefficient estimation for heavy-tailed autoregressive processes }\vskip15pt

\bigskip  Hailin Sang$^{a}$ and Yan Sun$^{b}$

\bigskip 

 $^{a}$ Department of Mathematics, The University of Mississippi \\University,  MS 38677, USA. Email: sang@olemiss.edu
 \bigskip
 
 $^{b}$ Department of Mathematics and Statistics, Utah State University \\Logan, UT 84322, USA. Email: yan.sun@usu.edu

\end{center}

\begin{abstract}
We propose a sparse coefficient estimation and automated model selection procedure for autoregressive (AR) processes with heavy-tailed innovations based on penalized conditional maximum likelihood. Under mild moment conditions on the innovation processes, the penalized conditional maximum likelihood estimator (PCMLE) satisfies a strong consistency, $O_P(N^{-1/2})$ consistency, and the oracle properties, where N is the sample size. We have the freedom in choosing penalty functions based on the weak conditions on them. Two penalty functions, least absolute shrinkage and selection operator (LASSO) and smoothly clipped average deviation (SCAD), are compared. The proposed method provides a distribution-based penalized inference to AR models, which is especially useful when the other estimation methods fail or under perform for AR processes with heavy-tailed innovations (see \cite{Resnick}). A simulation study confirms our theoretical results. At the end, we apply our method to a historical price data of the US Industrial Production Index for consumer goods, and obtain very promising results.

\end{abstract}

keywords: autoregressive process; causality; heavy tails; penalized maximum likelihood estimation; oracle properties; strong consistency\\


\section{Introduction}\label{intro}

The autoregressive (AR)($p$) process is one of the most fundamental time series models that have been extensively studied and applied in different fields. One major role that AR models play in the analysis of time series is the use of autoregressive representation of a stationary time series. While theoretically, such representation ``will give answers to many problems'' (\cite{Akaike69}), in practice, however, any AR process, being an approximation to what is observed in reality, must allow for an arbitrary magnitude of the order $p$, in order to achieve a satisfying approximation (see e.g. \cite{Poskitt07}). Autoregressive moving average (ARMA) process is one of such ``stationary time series'' that can be represented by an infinite order AR process. Inferences to the ARMA models are usually made by fitting a long-order AR model to the data, which is viewed as a truncation of the AR($\infty$) representation. See \cite{Shibata80}, \cite{GalbraithZinde-Walsh97}, \cite{GalbraithZinde-Walsh01}, \cite{IngWei05}, among others. Moreover, the need for long range dependency in the economic and financial data analysis also calls for the application of long-order AR processes. For instance, the autoregression-based approximation to the autoregressive fractionally integrated moving average (ARFIMA) processes is considered as an efficient and desirable method to make inferences of the long-memory ARFIMA models. See \cite{GalbraithZinde-Walsh01} 
and \cite{Poskitt07}, 
for a partial list of references.  Nevertheless, traditional model selection procedures based on criteria such as FPE \cite{Akaike73}, AIC \cite{Akaike70} and BIC \cite{Schwarz} are not efficient in fitting long order AR processes, especially when the AR process has a sparse structure.

In this paper, we propose an automated and efficient model selection procedure which is based on penalized conditional maximum likelihood for AR processes. The shrinkage estimators have a long history. See \cite{Stein81}, \cite{Ralescu92} and \cite{Brandwein93}, for examples. Technically, such estimators could obtain the shrinkage feature via the minimization of a loss function plus a penalty term, with the loss function being the least squares or the negative log likelihood in the usual cases. The existence of a suitably chosen penalty induces zero elements in the estimates, resulting in a simultaneous model selection procedure, while the parameters are being estimated. In the past two decades, a great deal of literature has been devoted to investigating such techniques, and a large number of penalty functions have been proposed including LASSO \cite{Tibshirani96}, SCAD \cite{FanLi01}, adaptive LASSO \cite{Zou06}. See \cite{FanLi02}, \cite{FanLi04}, \cite{HunterLi05}, \cite{YuanLin07}, and \cite{wuliu09} for a partial list of references. Although these techniques have been thoroughly studied and widely applied in the independent data settings, their performances in the time series context have not been studied very much.  Among the few existing relevant works, \cite{Huang07} considered a LASSO penalized least squares (PLS) method for a linear regression model with autoregressive errors, which was later extended in \cite{Yoon12} by allowing the penalty function to be chosen from LASSO, SCAD, and Bridge. \cite{Nardi11} studied the LASSO PLS for AR($p$) processes particularly under the \textquotedblleft double asymptotic\textquotedblright  framework, which means the order $p$ and the sample size go to infinity simultaneously. In all the aforementioned works the authors use the least squares method. 
In this paper, we propose a penalized sparse estimation for AR ($p$) models and thus develop a new model selection procedure. Based on the conditional likelihood, our PCMLE is especially useful when the time series model has heavy-tailed innovations. This is striking since the regular methods fail or under perform for AR processes with heavy-tailed innovations \cite{Resnick}.
Asymptotic properties of our PCMLE, regarding both estimation accuracy and model selection consistency, are investigated under the general conditional likelihood framework and mild conditions for the innovations and the penalty functions. 

Our theoretical results are two-fold. First, we give strong consistency of the PCMLE in Theorem \ref{strongconsistence} under weak conditions on the innovations and the penalty functions. In particular, we only require the sequence of penalty functions to be uniformly equicontinuous and converging to zero. The conditional maximum likelihood estimators with either LASSO or SCAD penalties enjoy the strong consistency. Second, we show that under certain regularity conditions on the innovations and the penalty functions, including the existence of the fourth moment of the innovations, the PCMLE of the coefficients are $N^{-1/2}$ consistent in probability. Furthermore, we derive the what have been known as ``oracle properties'' in the literature in Theorem \ref{thm2} for this $N^{-1/2}$ consistent PCMLE: 1) The coefficients whose true values are zero are estimated to be exactly zero with probability going to one. This property, referred to as sparsity, guarantees that the optimal model will be chosen with probability going to one. 2) The PCMLEs for the non-zero coefficients satisfy a multivariate central limit theorem, which states that asymptotically the estimated non-zero coefficients obtain the same efficiency as if the true sparse structure were known in advance. This immediately relaxes the constraint on the magnitude of the order $p$, as enlarging $p$ will no longer bring in proportionally more burden on the estimation efficiency. The PCMLE with SCAD penalty, but not with LASSO penalty, have the oracle properties. All these properties are confirmed by simulations with Gaussian and non-Gaussian innovations. Finally, we give a detailed discussion and rule of thumb on how the sample size should be adjusted, in order to minimize small sample risk and achieve optimal performances. 

In this paper, we shall use the following conventions: the notation $||\cdot||$ is used for the  $L_2$ norm; the notation $\Rightarrow$ denotes weak convergence; $X_n=o_P(1)$ is used for the convergence to zero in probability; and the bold face letters denote vectors. Besides, 
we denote by $X_n=O_P(1)$ a sequence of random variables $\{X_n\}$ bounded in probability (see e.g., Definition 3.3, \cite{Wooldridge02}). Throughout the paper, we assume the order $p$ is fixed, and does not increase with sample size N. 
 
The rest of the paper is organized as follows: Section 2 formally introduces our methodology and results. We discuss the performances of the PCMLE with two popular penalties, LASSO and SCAD, in Section 3. Simulation results are reported in Section 4, which include simulations with both Gaussian and non-Gaussian innovations. We demonstrate our method with a real data analysis in Section 5, which shows improved performances over the traditional MLE and FPE based model selection. We finish with a conclusion in Section 6. Proofs of our results are collected in Section 7. Useful lemmas and their proofs are deferred to the Appendix.

\section{Main results}
\label{results}
In this paper we study the PCMLE of the AR($p$) model
\begin{equation}\label{arp}
X_t=\phi_1X_{t-1}+\cdots+\phi_pX_{t-p}+Z_t.
\end{equation}
Let $\Theta$ be the space of parameter vectors $\boldsymbol{\theta}=(\phi_1,\cdots,\phi_p)^T$, $\boldsymbol{\theta_0}=(\phi_{1,0},\cdots,\phi_{p,0})^T$ be the underlying parameter vector, and $\sigma(X_{t-1},\cdots,X_{t-p})$ be the $\sigma$-algebra generated by the random variables $X_{t-1},\cdots,X_{t-p}$. 
Denote by $f_t(x):=f(x|\sigma(X_{t-1}, \cdots, X_{t-p});\boldsymbol{\theta})$ the conditional density function of $X_t$ given $X_{t-1},\cdots,X_{t-p}$.  Given observation\textbf{s} $X_1, \cdots, X_N$,  the conditional log likelihood function $L(\boldsymbol{\theta})$ is
\begin{eqnarray}
L(\boldsymbol{\theta})&:=&L(X_1, \cdots, X_N|\boldsymbol{\theta})\nonumber\\
&:=& \log \prod_{t=p+1}^N f_t(X_t)=\sum_{t=p+1}^N \log f_t(X_t):=\sum_{t=p+1}^N l_t(\boldsymbol{\theta}).\nonumber
\end{eqnarray}
Here we take the convention $\log 0=0$. As in the literature, the PCMLE of $\boldsymbol{\theta}$ is defined as
\begin{equation}\label{pmle}
\hat{\boldsymbol{\theta}}:=\hat{\boldsymbol{\theta}}_{\lambda_N}:=\text{argmax}_{\boldsymbol{\theta}\in\Theta}\{L(\boldsymbol{\theta})-NP_{\lambda_N}(\boldsymbol{\theta})\},
\end{equation}
where $P_{\lambda_N}(\boldsymbol{\theta})$ is a penalty function and $\lambda_N$ is a tuning parameter. Further, denote $$Q(\boldsymbol{\theta}^T):=Q(\boldsymbol{\theta}):=L(\boldsymbol{\theta})-NP_{\lambda_N}(\boldsymbol{\theta}).$$
We will make the following assumptions for all the results in this section.
\begin{assumptions} \label{as1}
\begin{enumerate}
\item The innovations $Z, \{Z_t\}_{-\infty}^{+\infty}$ are independent and identically distributed random variables (i.i.d.) with zero mean and variance $\sigma^2<\infty$.
\item $\Phi(z):= 1-\phi_{1,0}z-\cdots-\phi_{p,0}z^p\ne 0$ for all $z\in \mathbb{C}$ such that $|z|\le 1$.
\end{enumerate}
\end{assumptions}
Under the conditions from the first part, the second part of Assumptions \ref{as1} is equivalent to the causality of the time series AR($p$), i.e., there exists a sequence of constants $\{a_i\}$ such that $\sum_{i=0}^\infty |a_i|<\infty$ and  $X_t=\sum_{i=0}^\infty a_iZ_{t-i}$ (Theorem 3.1.1, \cite{BrockwellDavis87}). It is clear that this time series is weakly and strictly stationary with $EX_t=0$. Denote the autocovariance function by $\gamma(h)=Cov(X_t,X_{t+h})=EX_tX_{t+h}$.  \\
Let $g(z)$ be the density function of $Z$.  Observe that $f_t(x)=g(x-\sum_{j=1}^{p}\phi_jX_{t-j})$. Therefore, 
\begin{equation}\label{FtoG}
f_t(X_t)=g(X_t-\sum_{j=1}^{p}\phi_jX_{t-j})\;\;\text{and}\;\;l_t(\boldsymbol{\theta})=\log g(X_t-\sum_{j=1}^{p}\phi_jX_{t-j}).
\end{equation}
Especially,  
\begin{equation}\label{Ltheta0}
l_t(\boldsymbol{\theta}_0)=\log g(X_t-\sum_{j=1}^{p}\phi_{j,0}X_{t-j})=\log g(Z_t).
\end{equation}
 The next theorem gives the conditions such that the PCMLE $\hat{\boldsymbol{\theta}}_{\lambda_N}$ has strong consistency. 
\begin{theorem}\label{strongconsistence}
Assume that the parameter vector space $\Theta$ is compact, $g(z)$ is continuous and $E|\log g(Z)|<\infty$. Further, we assume that $\{P_{\lambda_N}(\boldsymbol{\theta})\}$ are uniformly equicontinuous in $\Theta$ and $P_{\lambda_N}(\boldsymbol{\theta})\rightarrow 0$ as $N\rightarrow \infty$ for each $\boldsymbol{\theta}\in\Theta$. Then under Assumptions \ref{as1}, $\hat{\boldsymbol{\theta}}_{\lambda_N}$ converges to $\boldsymbol{\theta}_0$ almost surely.
\end{theorem}
Usually the vector $\boldsymbol{\theta_0}=(\phi_{1,0},\cdots,\phi_{p,0})^T$ has some zero components. Without loss of generality, we assume that the underlying parameter vector $\boldsymbol{\theta_0}=(\phi_{1,0},\cdots,\phi_{p,0})^T$ has $s$ zeros and these zeros are the first $s$ parameters. Then we write 
$$\boldsymbol{\theta}^T_0=(\phi_{1,0},\cdots,\phi_{p,0})=(0,\cdots, 0, \phi_{s+1,0},\cdots,\phi_{p,0}):=(\boldsymbol{0}^T, \boldsymbol{\theta}^T_{0,1}):=(\boldsymbol{\theta}^T_{0,0}, \boldsymbol{\theta}^T_{0,1}).$$
With the same rearrangement,  $\boldsymbol{\theta}^T=(\phi_1,\cdots,\phi_p):=(\boldsymbol{\theta}^T_{1,0}, \boldsymbol{\theta}^T_{1,1})$.\\
The results in the rest of this section need the following extra assumptions. 
\begin{assumptions} \label{as2}
\begin{enumerate}
\item $Z$ has a finite fourth moment.
\item $E\frac{ (g''(Z))^2}{ g^2(Z)}<\infty$ and $E\frac{(g'(Z))^4}{ g^4(Z)}<\infty$.
\item $\left(\frac{g'}{g}\right)''(z)<B$ uniformly for some constant $B$. 
\end{enumerate}
\end{assumptions}
Besides, $E\frac{(g'(Z))^2}{ g^2(Z)}<\infty$ if $E\frac{(g'(Z))^4}{ g^4(Z)}<\infty$.  We denote 
\begin{equation}\label{cg}
C(g):=E\frac{(g'(Z))^2}{ g^2(Z)}.
\end{equation}
\begin{example}
In the important case $Z\sim N(0,1)$, $g'(z)=-zg(z)$ and $g''(z)=(z^2-1)g(z)$. Therefore, $C(g)=E\frac{(g'(Z))^2}{ g^2(Z)}=EZ^2=\sigma^2<\infty$. $E\frac{ (g''(Z))^2}{ g^2(Z)}=E(Z^2-1)^2<\infty$ . $\left(\frac{g'}{g}\right)''(z)=0$. $E\frac{(g'(Z))^4}{ g^4(Z)}=EZ^4<\infty$. We only require the existence of the fourth moment of the innovation in the assumptions. Therefore our results are good for AR processes with heavy tails also. For example, the t distributions with degree of freedom $df>4$ satisfy all the conditions in Assumptions \ref{as2}. The algebra is tedious but routine. 
\end{example}
\noindent In the following propositions and theorems, the penalty function $P_{\lambda_N}(\boldsymbol{\theta})$ has the form
$$P_{\lambda_N}(\boldsymbol{\theta})=\sum_{i=1}^p p_{\lambda_N}(|\phi_i|).$$
\begin{assumptions}\label{as3}
The assumptions on the penalty $p_{\lambda_N}(|\phi|)$ are
\begin{enumerate}
\item $\lambda_N\rightarrow 0$, $\sqrt{N}\lambda_N\rightarrow\infty$ as $N\rightarrow \infty$ and 
$\liminf_{N\rightarrow\infty}\liminf_{\phi\rightarrow 0^+} p'_{\lambda_N}(|\phi|)/\lambda_N>0;$
\item $p_{\lambda_N}(\phi)\ge 0$, $p_{\lambda_N}(0)=0$, $a_N=\max\{|p_{\lambda_N}'(|\phi_{i,0}|)|: \phi_{i,0}\ne 0\}\rightarrow 0$, $\max\{|p_{\lambda_N}''(|\phi_{i,0}|)|: \phi_{i,0}\ne 0\}\rightarrow 0$ as $N\rightarrow\infty$ and $p_{\lambda_N}'''$ exists and is bounded. 
\end{enumerate}
\end{assumptions}
\begin{proposition}\label{sparsity}
Assume Assumptions \ref{as1}, \ref{as2} and part 1 of Assumptions \ref{as3}. With probability tending to $1$, for any given $\boldsymbol{\theta}_{1,1}$ with $||\boldsymbol{\theta}_{1,1}-\boldsymbol{\theta}_{0,1}||=O_p(N^{-1/2})$, we have
$$Q(\boldsymbol{0}^T, \boldsymbol{\theta}^T_{1,1})^T=\max_{||\boldsymbol{\theta}_{1,0}||\le CN^{-1/2}}Q(\boldsymbol{\theta})$$
for some constant $C$.
\end{proposition}
This proposition gives the sparsity of the $N^{-1/2}$ consistent estimator, i.e., the coefficients whose true values are zero are estimated to be exactly zero with probability tending to $1$. The next proposition is useful to provide the $N^{-1/2}$ consistent PCMLE.
\begin{proposition}\label{weakconsistence}
Assume Assumptions \ref{as1}, \ref{as2} and part 2 of Assumptions \ref{as3}. Then there exists a local maximizer $\hat{\boldsymbol{\theta}}$ of $Q(\boldsymbol{\theta})$ such that $||\hat{\boldsymbol{\theta}}-\boldsymbol{\theta}_0||=O_P(N^{-1/2}+a_N)$, with probability going to one, where $a_N$ is defined as in Assumptions \ref{as3}.   
\end{proposition}
Under Assumptions \ref{as3}, if the quantity $a_N$ defined in Assumptions \ref{as3} satisfies $a_N=O(N^{-1/2})$, the local maximizer $\hat{\boldsymbol{\theta}}$ of $Q(\boldsymbol{\theta})$ in Proposition \ref{weakconsistence} is a $N^{-1/2}$ consistent PCMLE of $\boldsymbol{\theta}$. Therefore, 
from Proposition \ref{sparsity}, this estimator has sparsity, i.e., with probability tending to $1$ the estimates of the zero coefficients are zeros, $\hat{\boldsymbol{\theta}}_{1,0}=\textbf{0}$. We list this conclusion as the first part of the following theorem. Further,  we show that the estimates of the non-zero coefficients satisfy an asymptotic normality in the second part of this theorem.  

\begin{theorem}\label{thm2}
Let $\Gamma$ be the non-negative definite ${(p-s)\times (p-s)}$ matrix with the entry $\Gamma(l,m)=\gamma(m-l)$, $1\le l, m\le p-s$.
Denote $$\Delta=\text{diag}\left\{p''_{\lambda_N}(|\phi_{s+1,0}|),\cdots, p''_{\lambda_N}(|\phi_{p,0}|)\right\},$$
$$\boldsymbol{b}=(p'_{\lambda_N}(|\phi_{s+1,0}|)sgn(\phi_{s+1,0}),\cdots,p'_{\lambda_N}(|\phi_{p,0}|)sgn(\phi_{p,0}))^T.$$
Assume $a_N=O(N^{-1/2})$. Under assumptions \ref{as1}, \ref{as2} and \ref{as3}, the local  maximizer $\hat{\boldsymbol{\theta}}=(\hat{\boldsymbol{\theta}}_{1,0}^T, \hat{\boldsymbol{\theta}}_{1,1}^T)^T$ of $Q(\boldsymbol{\theta})$ satisfies 
\begin{enumerate}
\item $\hat{\boldsymbol{\theta}}_{1,0}=\textbf{0}$,
\item $\sqrt{N}[(C(g)\Gamma+\Delta)(\hat{\boldsymbol{\theta}}_{1,1}-\boldsymbol{\theta}_{0,1})+\boldsymbol{b}]\Rightarrow N(0, C(g)\Gamma).$
\end{enumerate}
Here the constant $C(g)$ is given by (\ref{cg}).
\end{theorem}
\begin{remark}

From the statistical inference point of view, we need to have estimations for the matrix $\Gamma$ and the constant $C(g)$ if the innovation density function $g(x)$ is unknown. Each of these estimations is an important independent research topic itself and possesses a great deal of works in the literature.  For the $(p-s)\times (p-s)$ fixed-dimensional non-negative matrix $\Gamma$, its entries $\gamma(h)$, or $\gamma(-h)$, are estimated consistently by the sample auto covariance function $\hat\gamma(h)=\frac{1}{N}\sum_{i=1}^{N-h}X_iX_{i+h}$, $0\le h\le p-s$.  See \cite{BrockwellDavis87} for more details. To keep the consistency while allowing the dimension of $\Gamma$ to grow with the sample size, banding or tapering is implemented. See  \cite{BickelLevina08}, \cite{XiaoWu12} and the references therein for the new development in this field. We have assumed that the innovation density $g(x)$ and therefore the constant $C(g)$ are known in the above theory we develop. In the case that $C(g)$ is unknown, in the literature one uses the residuals after fitting the AR model to estimate the innovation density function $g(x)$, its derivative $g'(x)$, and therefore $C(g)$, by standard density estimation methods such as the kernel method.  For the residual based innovation density estimation, see \cite{Robinson87}, \cite{Liebscher99}, \cite{MullerSchick05} and the references therein. 
\end{remark}
\section{Discussion}
\label{discussion}
In this section we discuss two popular penalties, SCAD (\cite{FanLi01}) and LASSO (\cite{Tibshirani96}). The Smoothly Clipped Average Deviation (SCAD) is defined by its first derivative as follows:
\begin{equation}\label{scad'} 
  p_{\lambda_N}'(|\phi|)=\lambda_N I(|\phi|\leq\lambda_N)+\frac{(a\lambda_N-|\phi|)_{+}}{a-1}I(|\phi|>\lambda_N),
\end{equation}
where  $a>2$ is the second tuning parameter. More precisely, 
\begin{equation}\label{scad}
\begin{split}
  p_{\lambda_N}(|\phi|) &=\lambda_N|\phi|I(|\phi|\leq\lambda_N)
  +(\frac{a\lambda_N}{a-1}|\phi|-\frac{|\phi|^2}{2(a-1)}-\frac{\lambda_N^2}{2(a-1)})I(\lambda_N<|\phi|<a\lambda_N)\\
  &+\frac{(a+1)\lambda_N^2}{2}I(|\phi|\geq a\lambda_N).
\end{split}
\end{equation} Further, 
\begin{equation}\label{scad''}
p_{\lambda_N}''(|\phi|)=-(a-1)^{-1}I(\lambda_N<|\phi|<a\lambda_N).
\end{equation}
The Least Absolute Shrinkage and Selection Operator (LASSO) is defined as the absolute value of the parameter  with a scaling parameter $\lambda_N$. That is, $p_{\lambda_N}=\lambda_N|\phi|$. For both LASSO and SCAD penalty, $\lambda_N\ge 0$.\\
\indent To have the  strong consistency as in Theorem \ref{strongconsistence},  we require that $\{P_{\lambda_N}(\boldsymbol{\theta})\}$ be uniformly equicontinuous and $P_{\lambda_N}(\boldsymbol{\theta})\rightarrow 0$ as $N\rightarrow \infty$ for each $\boldsymbol{\theta}\in\Theta$. This condition is satisfied simply by setting $\lambda_N\to 0$ for both SCAD and LASSO penalties. Therefore, PCMLE with either LASSO or SCAD penalty enjoys the strong consistency. \\
\indent Recall the definition of $a_N$ in Assumptions \ref{as3}. For LASSO, $a_N=\lambda_N$. To have Assumptions \ref{as3}, we need $N^{1/2}a_N\rightarrow \infty$. This is a contradiction with $a_N=O(N^{-1/2})$. Therefore, from what we have proved, there is not enough evidence to claim that the PCMLE (\ref{pmle}) with LASSO penalty has the oracle properties.  For SCAD, it is easy to verify that $p_{\lambda_N}(|\phi|)\ge 0$, $p_{\lambda_N}(0)=0$ by (\ref{scad}). From (\ref{scad'}), we have 
$$\liminf_{N\rightarrow\infty}\liminf_{\phi\rightarrow 0^+}p'_{\lambda_N}(|\phi|)/\lambda_N=\liminf_{N\rightarrow\infty} 1=1>0.$$ 
From (\ref{scad''}),  
\begin{equation}\label{max}
\max\{|p_{\lambda_N}''(|\phi_{i,0}|)|:\phi_{i,0}\ne 0\}=(a-1)^{-1}I(\lambda_N<|\phi_{i,0}|<a\lambda_N \;\;\text{for some}\; i)
\end{equation}
and $p_{\lambda_N}'''=0$. Besides, 
\begin{eqnarray}
a_N&=&\max\{|p_{\lambda_N}'(|\phi_{i,0}|)|: \phi_{i,0}\ne 0\}\notag\\
&=& \max\{\lambda_NI(\min_{|\phi_{i,0}|\ne 0}|\phi_{i,0}|\le \lambda_N), \max \frac{a\lambda_N-|\phi_{i,0}|}{a-1}I(\lambda_N<|\phi_{i,0}|<a\lambda_N)\}.\notag
\end{eqnarray}
Therefore, $a_N=0$ if $\min_{\phi_{i,0}\ne 0} |\phi_{i,0}|\ge a\lambda_N$. Otherwise, $a_N=O(\lambda_N)$. Hence, for the sequence $\{\lambda_N\}$ with  $\lambda_N\rightarrow 0$ and $N^{1/2}\lambda_N\rightarrow\infty$, $(\ref{max})=a_N=0$ if $N$ satisfies $\min_{\phi_{i,0}\ne 0} |\phi_{i,0}|\ge a\lambda_N$. But $\min_{\phi_{i,0}\ne 0} |\phi_{i,0}|\ge a\lambda_N$ is true eventually if $\lambda_N\rightarrow 0$. So the PCMLE with SCAD penalty has the oracle properties if $\lambda_N\rightarrow 0$ and $N^{1/2}\lambda_N\rightarrow\infty$.  
In practice, it is recommended to choose sample size $N$ with $\min_{\phi_{i,0}\ne 0} |\phi_{i,0}|\ge a\lambda_N$ after the sequence $\{\lambda_N\}$ is selected if one has the information on $\min_{\phi_{i,0}\ne 0} |\phi_{i,0}|$, which can be routinely obtained by traditional estimations like least squares or MLE.  $\lambda_N$ should be selected with $N^{1/2}\lambda_N\rightarrow\infty$ but can be close to $N^{-1/2}$. If so, sample size $N$ should be chosen with  $N^{-1/2}=o(\min_{\phi_{i,0}\ne 0} |\phi_{i,0}|)$ but possibly close to $(\min_{\phi_{i,0}\ne 0} |\phi_{i,0}|)^{-2}$.  

\section{Simulation study}
\label{simu}
In this section we look at the performances of the two penalties, SCAD  and LASSO,  by numerical experiments.
The simulations are two-fold. On one hand, we simulate data from AR($p$) models which contain only zero and ``large'' non-zero parameters. The non-zero parameters are ``large'' in the sense that they are well above the order of $O(N^{-1/2})$, and therefore have very little risk to be mistakenly shrunk to 0 by the penalty. The performances of MLE and PCMLE are compared. On the other hand, we also consider the cases when some ``small'' non-zero parameters are involved in the model. That is, some non-zero parameters are smaller than $O(N^{-1/2})$.  Just as we expected, the numerical results show no statistical difference between the zero parameter and the ``small'' non-zeros.

We get a preliminary estimate of the coefficients by the usual MLE, which is next used as the initial value for the PCMLE algorithm. In the literature, there are two algorithms to compute the PCMLE, both of which are based on polynomial approximations of the penalty functions and eventually lead to a modified Newton-Raphson algorithm. The earlier one is the Local Quadratic Approximation (LQA) method proposed in \cite{FanLi01}. This algorithm essentially iteratively uses the Ridge penalty which does not produce zero estimates \cite{FrankFriedman93}. In practice the zeros are picked out heuristically rather than by the algorithm. This is a drawback since a parameter stays at zero after it is determined to
be zero at some iteration. A later improvement for the LQA is the Local Linear Approximation (LLA) method proposed in \cite{Zou08}, which iteratively computes the PCMLE with sparsity. Furthermore, the employment of the LLA offers the convenience to take advantage of many standard LASSO algorithms, by which LLA is computationally much more efficient than LQA. We therefore choose LLA to compute our PCMLE, specifically, the one-step LLA sparse estimator proposed in \cite{Zou08}. The first 80 percent of the sample is used to compute the PCMLE of the coefficients, and the tuning parameters are chosen by maximizing the unpenalized likelihood on the remaining 20 percent of the sample. 

We first simulate the following AR(5) model with  sample size N=1000,
\begin{equation}\label{model_1}
  X_t=0.2X_{t-1}+0.2X_{t-3}+0.2X_{t-5}+Z_t.
\end{equation}
Here, the innovation process $\left\{Z_t\right\}$ is generated from standard normal distribution independently. Notice that $0.2$ is well above the threshold $O(N^{-1/2})$. It is also easy to verify that such a combination of coefficients satisfies part 2 of Assumptions \ref{as1}, the causality. Table \ref{tab:AR5_1} reports a detailed result comparing the performances of MLE, LASSO PCMLE, and SCAD PCMLE. The error refers to the $L_2$ norm of the difference between the estimated coefficients and their true values. The std refers to the standard error calculated by the sandwich formula \cite{FanLi01}. It is clear that the SCAD PCMLE detects the zero coefficients. The LASSO PCMLE fails to identify one zero coefficient. In addition, SCAD PCMLE has improved estimation errors and standard errors. 

We repeat the above process with $N=1000$ for 100 times independently, and the results are summarized in Table \ref{tab:AR5_100}. The probability to be identified as 0 is calculated by the sample portion of the 100 trials for each coefficient. The average bias is the absolute difference between the mean value of the 100 estimates and the corresponding true value. The LASSO PCMLE is relatively conservative in terms of sparsity. Consequently, approximately for only 1/3 of the 100 times does the LASSO PCMLE correctly identifies each of the zero coefficients. In comparison, this proportion increases to approximately 3/4 for SCAD PCMLE. Especially, the sample probabilities of correctly getting 2 zeros for LASSO and SCAD PCMLE are 0.2 and 0.61 respectively. The observed biases of SCAD PCMLE are also smaller than those of LASSO PCMLE. We calculate the sample probability to get 0 estimate out of 100 independent trials for the two zero coefficients $\phi_{2,0}$, $\phi_{4,0}$, respectively and simultaneously, at sample size N=1000, 1500, 2000, 2500, 3000, 3500, 4000, and draw it as a function of N in Figure 1. For SCAD PCMLE, the probability for each of the two zero coefficients increases from around 0.7 to almost 0.95, as sample size grows from 1000 to 4000. Whereas for LASSO PCMLE, this proportion mostly varies between 0.3 to 0.4, and does not increase significantly as sample size increases. The contrast is even sharper when looking at the probability of getting both zeros. For SCAD PCMLE, it increases from 0.51 (N=1000) to 0.9 (N=4000). However, for LASSO PCMLE, it merely fluctuates around 0.2, never reaching 0.3.

\begin{figure}[ht]
\centering
\includegraphics[ height=1.800in, width=2.400in]{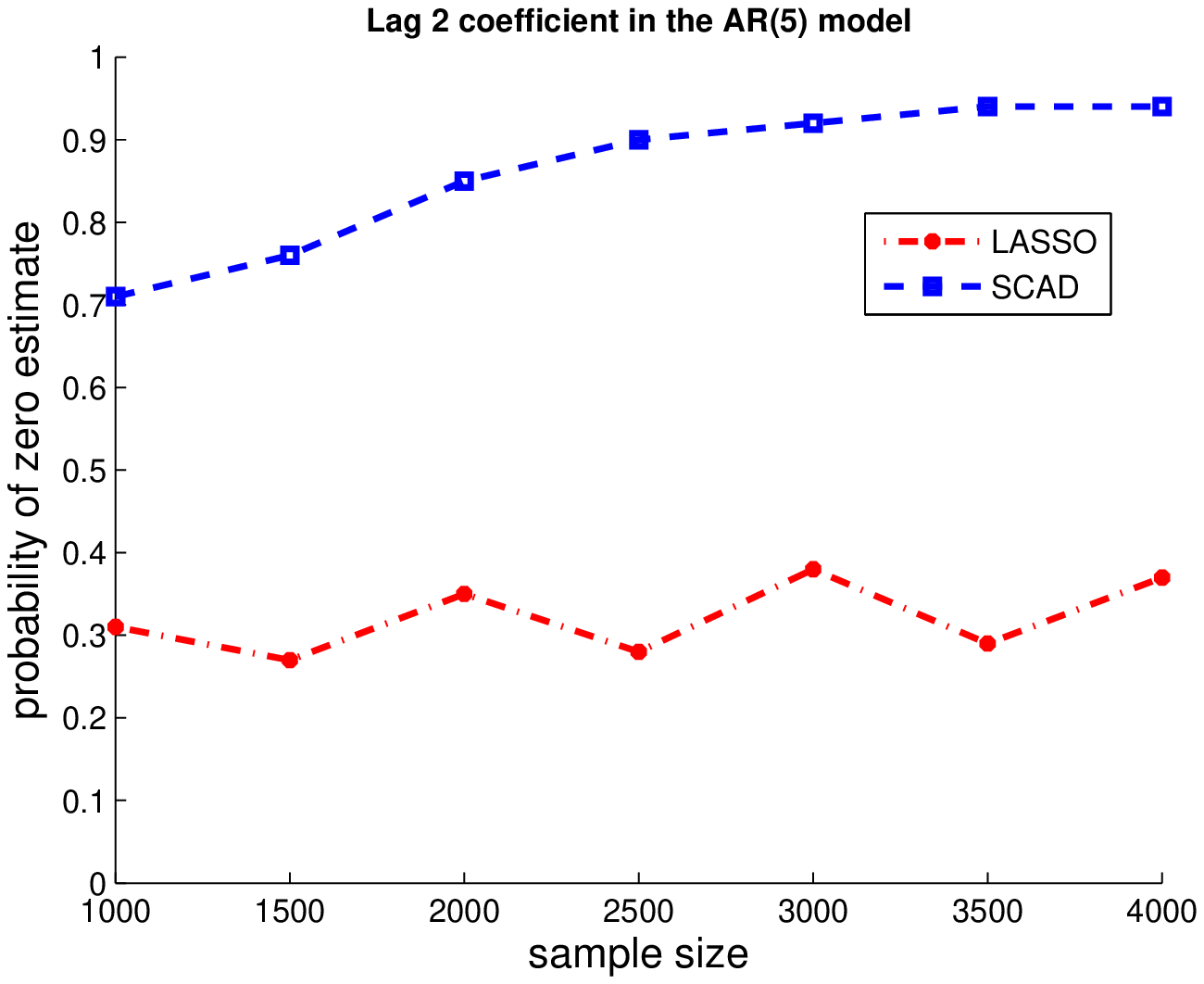}
\includegraphics[ height=1.800in, width=2.400in]{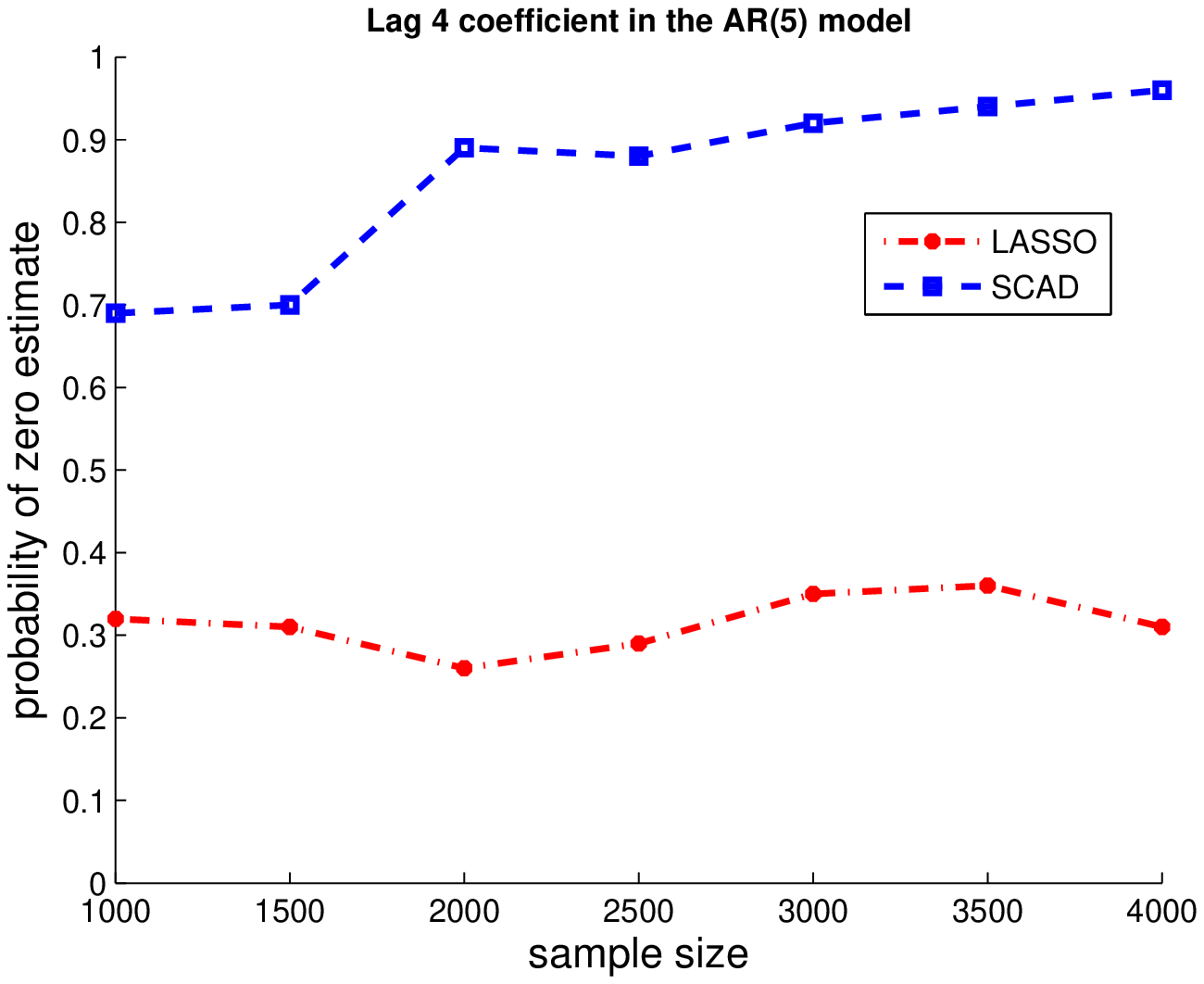}\\
\includegraphics[ height=1.800in, width=2.400in]{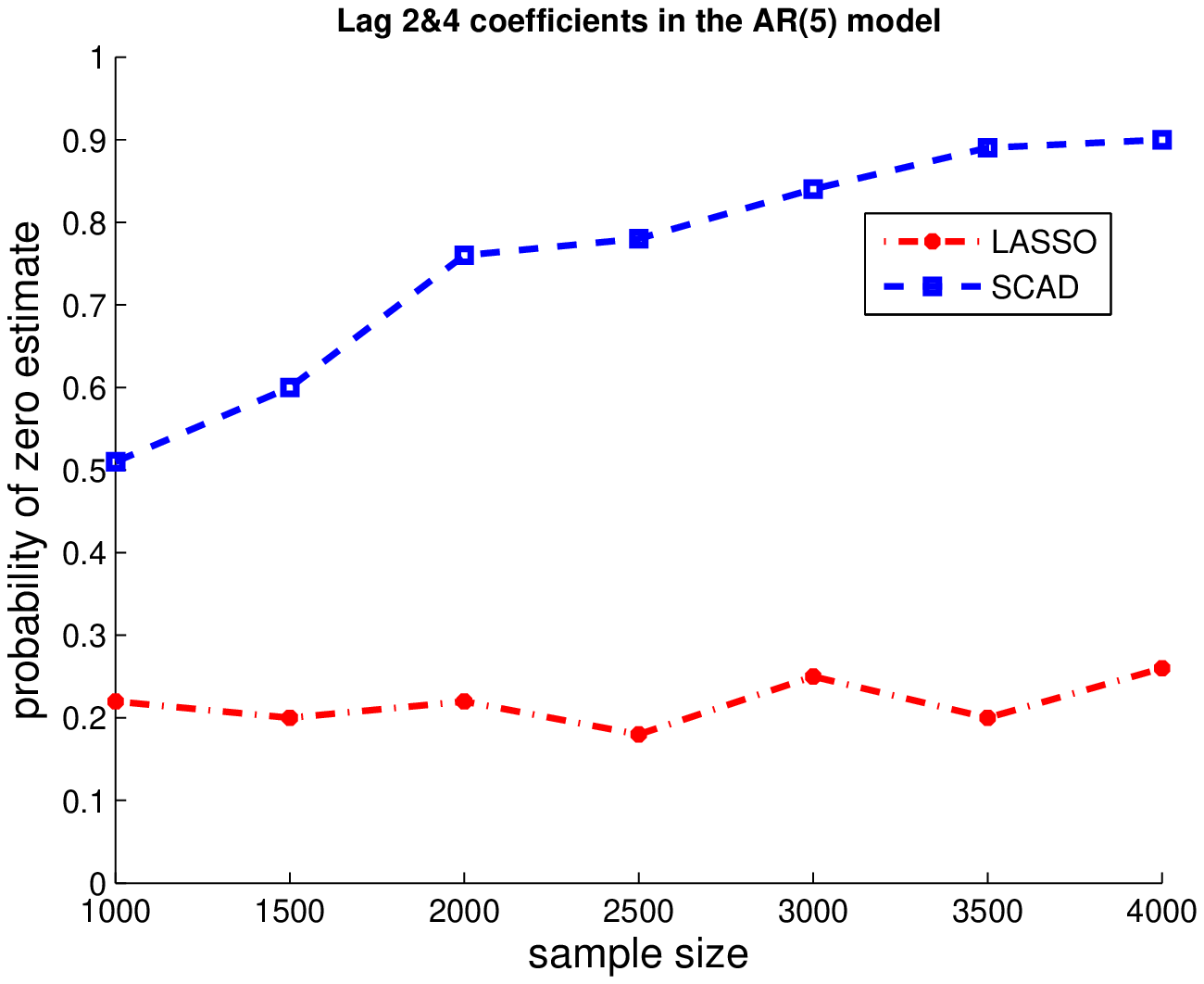}
\caption{The probability of zero estimates as a function of sample size for: 1) $\phi_{2,0}$ (upper left), 2) $\phi_{4,0}$ (upper right), 3) $\phi_{2,0}$ and $\phi_{4,0}$ simultaneously (lower), in model (\ref{model_1}).}
\label{fig:prob_model1}
\end{figure}

\indent We further consider the following model:
\begin{equation}\label{model_3}
    X_t=0.2X_{t-1}+N^{-3/4}X_{t-3}+\frac{1}{2}N^{-3/4}X_{t-5}+Z_t.
\end{equation}
That is, $\phi_{3,0}$ and $\phi_{5,0}$ now have order $O(N^{-3/4})$ for some fixed $N$, which is smaller than $O(N^{-1/2})$. By the foregoing discussion, they may not be detectable from the non-zeros by the PCMLE. Same as before, we carry out 100 independent experiments to estimate the coefficients in model (\ref{model_3}) using LASSO/SCAD PCMLE. Consistently, there is no problem with $\phi_{1,0}$. It is well above 0, and both LASSO and SCAD penalties distinguish it from 0 for the 100 experiments. Figure 2 plots the sample probability of zero estimates as a function of sample size for the other four coefficients. Notice that, statistically, there is no more difference between the two non-zero coefficients $\phi_{3,0}$, $\phi_{5,0}$ and $\phi_{2,0}$, $\phi_{4,0}$ shown in the plots. The four plots, referring to the four coefficients, look almost identical. 

\begin{figure}
\centering
\includegraphics[ height=1.500in, width=2.300in]{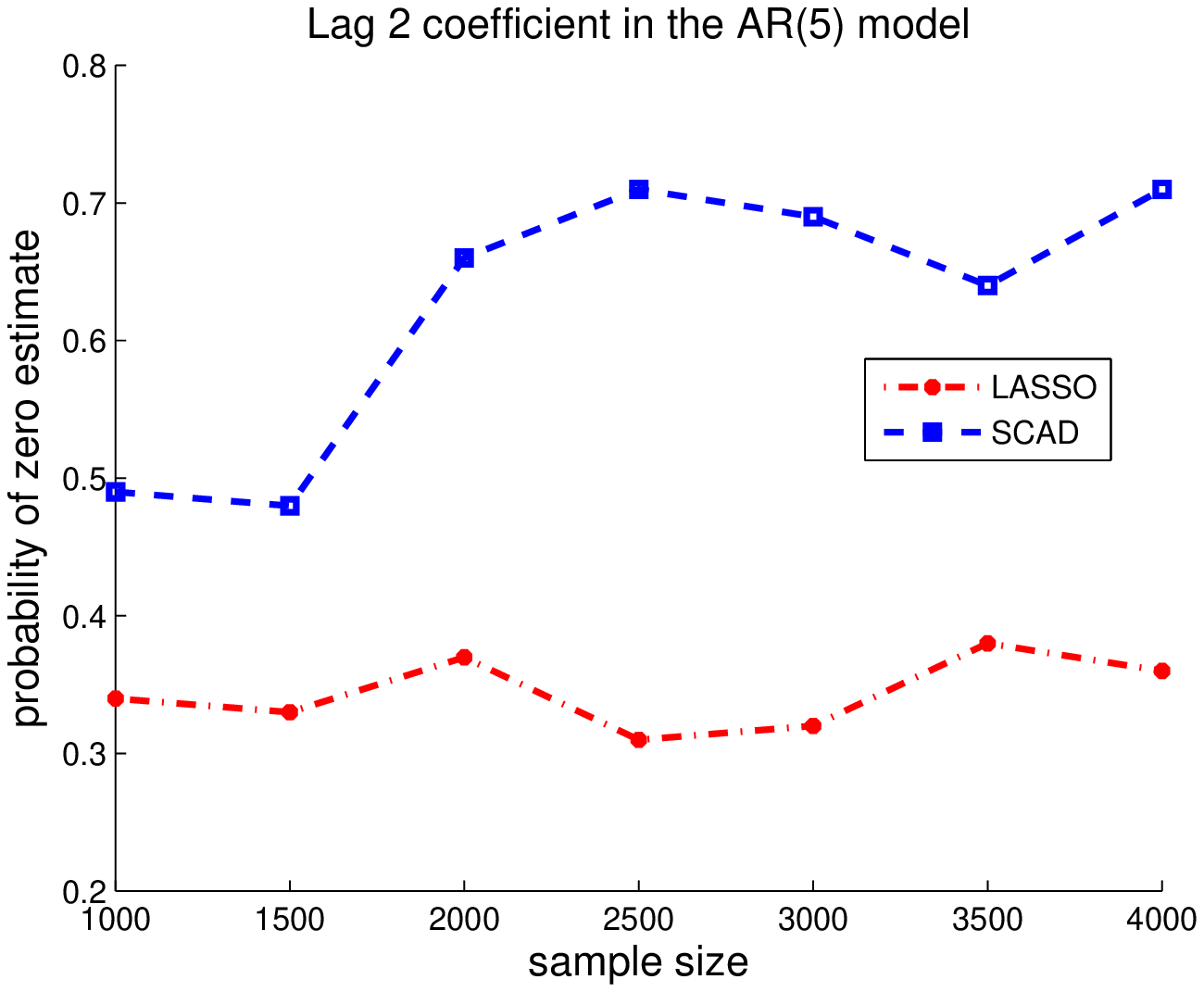}
\includegraphics[ height=1.500in, width=2.300in]{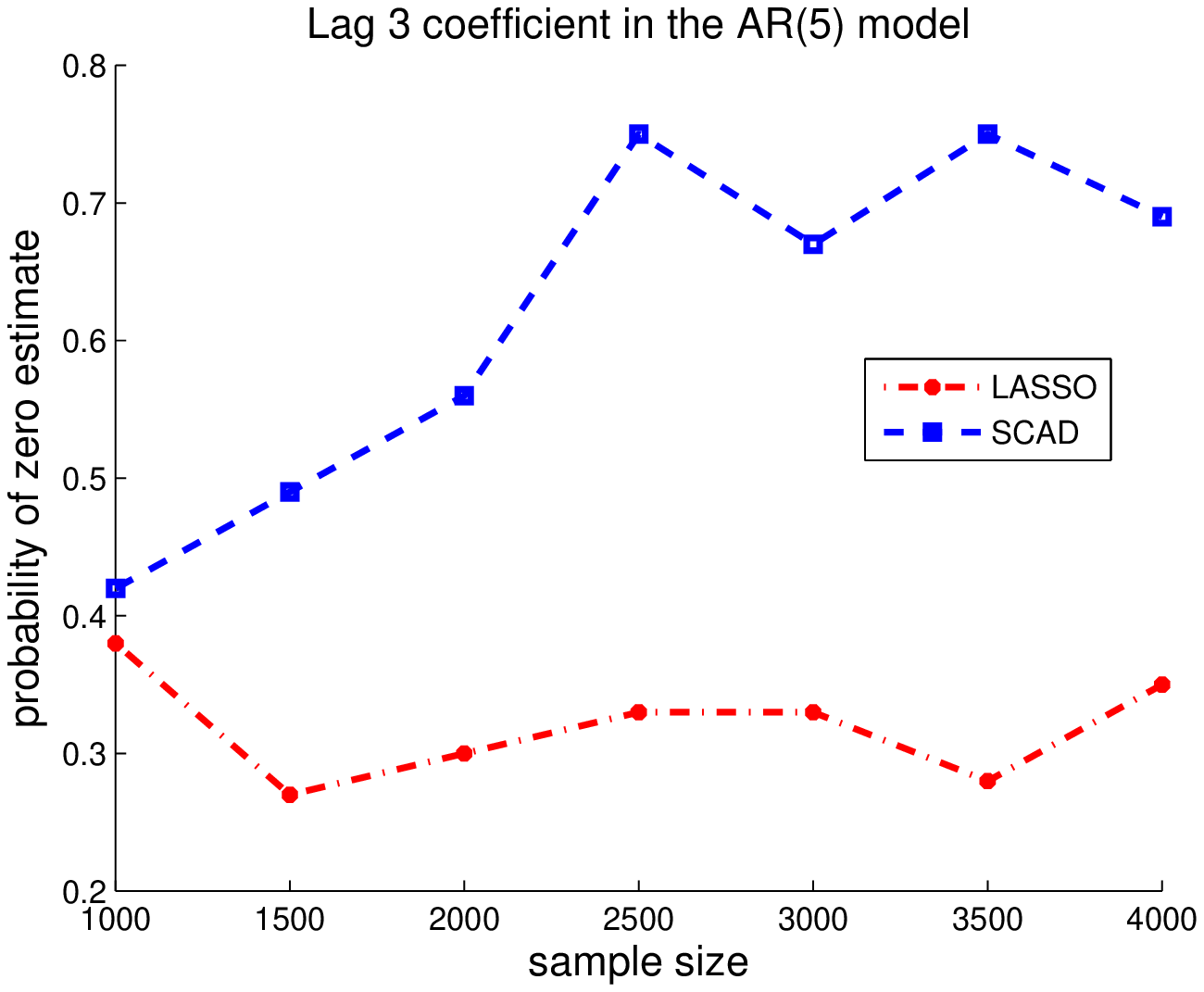}\\
\includegraphics[ height=1.500in, width=2.300in]{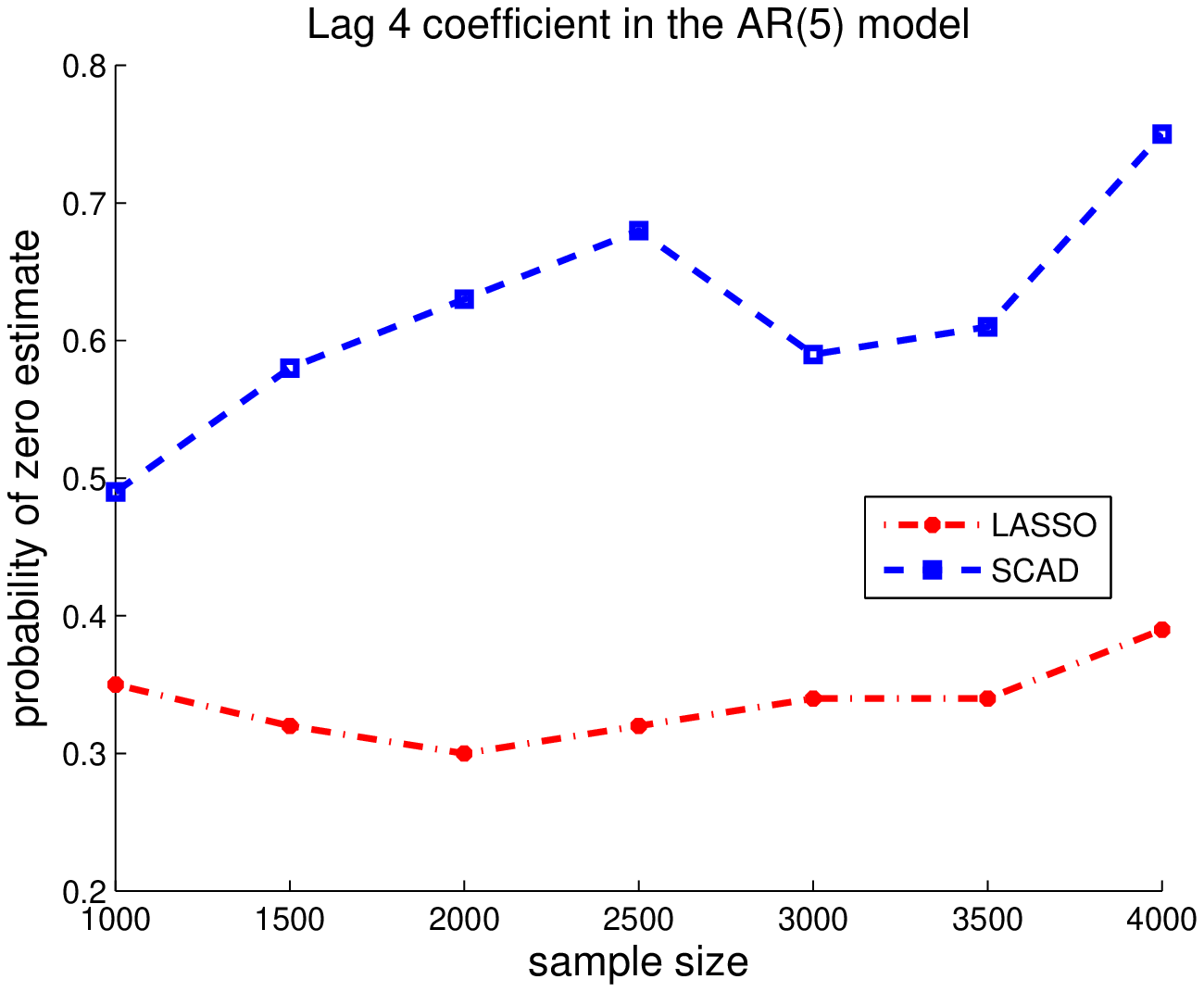}
\includegraphics[ height=1.500in, width=2.300in]{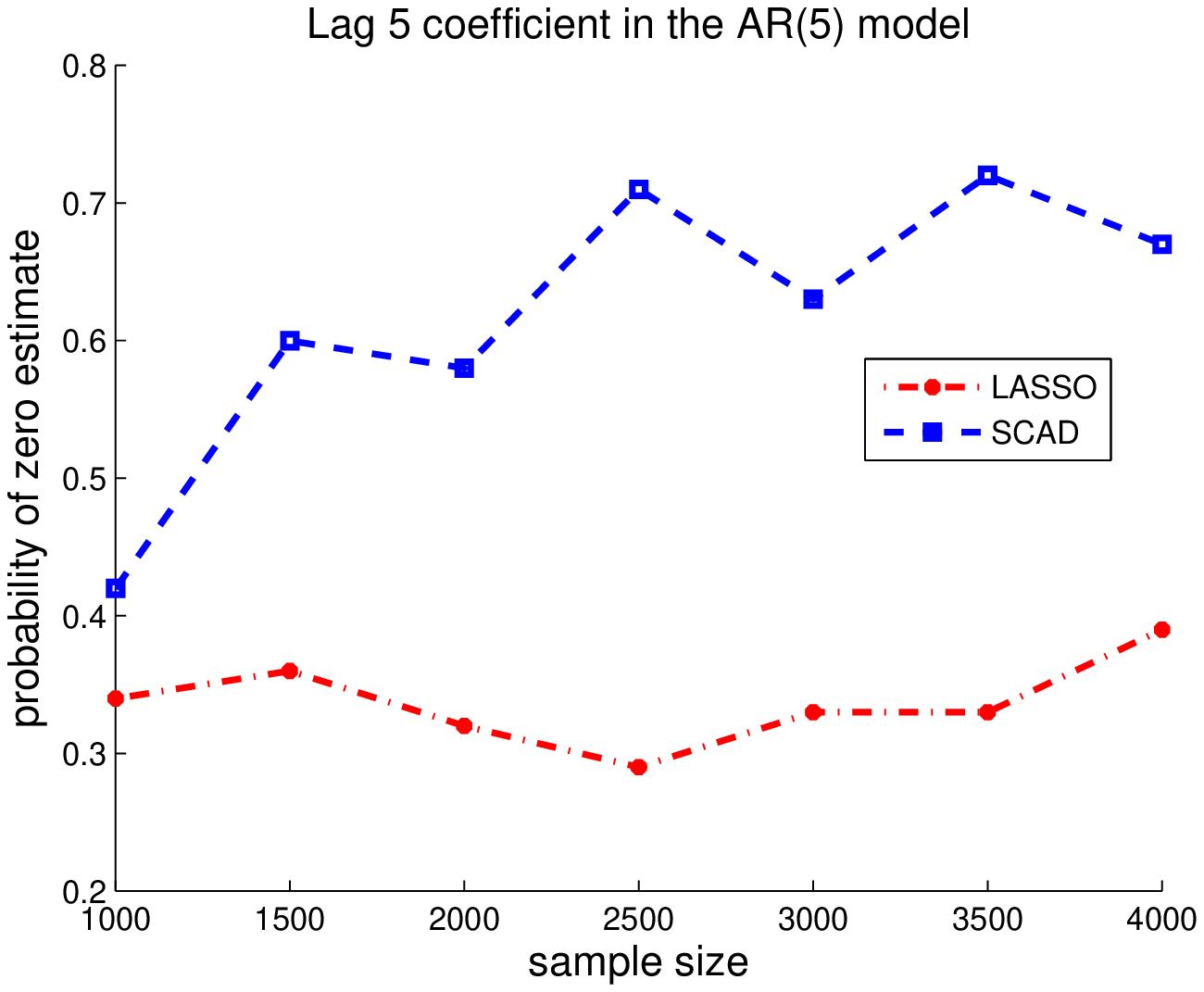}\\
\caption{The probability of zero estimates for the four coefficients 
 in model (\ref{model_3}).}
\label{fig:prob_model3}
\end{figure}
Finally, we consider models with student t innovations. It is easy to check that the density of t distribution with degree of freedom greater than 4 satisfies all the conditions in Theorems \ref{strongconsistence} and  \ref{thm2}. Therefore the PCMLE is expected to perform as well as that for the normal innovations. We simulate samples with length $N=1000$, and the degree of freedom of the T distribution $df=2, 5$. The estimation results of MLE, LASSO PMLE and SCAD PMLE are presented in Table \ref{tab:AR_T2}. The error refers to the $L_2$ norm of the difference vector between the estimated coefficients and their true values. 
For $df=2$, when the condition of Theorem \ref{thm2} is not satisfied, the errors of SCAD PMLE are even higher than those of MLE.  

\section{Application to real data}
In this section, we apply the penalized conditional likelihood method to analyze the US Industrial Production Index for consumer goods from January 1939 to August 2010 (www.economagic.com). The dataset consists of totally 860 seasonally adjusted monthly observations. We use the first 800 observations for in-sample estimation, and the last 60 for out-of-sample forecast. The first order differencing is applied to the original series to get rid of the linear trend. We fit three AR($p$) models ($p=20, 25, 30$) using both the MLE and the SCAD  PMLE. An AR($p$) model with an optimal order $p=24$ chosen by the Final Prediction Error (FPE) criterion \cite{Akaike70} is also included in the comparison. After the model is fitted, the differencing is converted and all forecast values are constructed for the original series.

We use two criteria, Mean Absolute Error (MAE) and Root Mean Square Error (RMSE), 
to evaluate the forecasts. In this example, we choose forecast steps $k=1, 6, 12$.  Let $m$ denote the total number of forecasts during the period for which the actual value $X(t)$ is known, and $F(t)$ denote the forecast value. Then, as in the literature, the MAE and RMSE are defined as: 
\begin{equation}
  MAE=\sum_{s=0}^{m-k}\frac{|F(N+s+k)-X(N+s+k)|}{m*|X(N+s)|}, \nonumber
\end{equation}  
\begin{equation}
  RMSE=\sum_{s=0}^{m-k}\left\{\frac{[F(N+s+k)-X(N+s+k)]^2}{m*X(N+s)^2}\right\}^{1/2}. \nonumber
\end{equation}
The comparative results for the forecasts of all the combinations of models and methods are summarized in Table \ref{tab:real2_forc}. The forecast errors of the SCAD PMLE are consistently smaller than those of the MLE for all the three models considered. The forecasting performances of the associated AR(24) model chosen by FPE are also shown here for the purpose of comparison.
All the three AR models fitted by SCAD PCMLE, AR($20$),   AR($25$),   AR($30$), have sparsity. Whereas the AR models fitted by regular MLE or the AR(24) model selected by FPE do not have any zero estimates at all. 
The coefficient estimates  from all the models and methods are listed in Table \ref{tab:real2_esti}. As seen from this table, lag 24 is very significant. This is why the FPE chooses 24 as the best order. However, the cost of choosing such a long-order AR model has obviously resulted in a poor prediction accuracy as can be seen from Table \ref{tab:real2_forc}. The SCAD PMLE picks up only 6 significant lags: $2, 3, 9, 18, 23, 24$, which has helped improve the prediction accuracy significantly. Also, the forecasting errors of SCAD PMLE are quite stable, except that for $p=20$ the forecasting errors are relatively higher, because one significant lag, lag 24, is excluded from the model.

\section{Conclusion}
In this paper, we propose a new sub-model selection procedure for AR($p$) models based on penalized maximum likelihood estimators of the coefficients. We prove that the resulting sparse PCMLE for the coefficient profile is both strongly consistent and locally $N^{-1/2}$ consistent under mild conditions. More importantly, under slightly additional conditions, we establish an oracle properties for the sparse estimator, analogous to the one by \cite{FanLi01} for independent observations. It says that the zero coefficients are estimated to be exactly zero with probability going to one, and the estimates for the non-zero ones are estimated as efficiently as if the true sub-model were known in prior. This property, together with the overall consistency, guarantees that the optimal sub-model is selected with probability tending to one, and the estimation efficiency for the selected coefficients gets improved by reducing from the full model to the sub-model. What is the most important, these are all done by running the model once, saving a great deal of computational cost from traditional sub-model selection methods. 

Although the asymptotic theorems look ideal, finite sample performances could be very different and even misleading. In order to give more guidance for practical use of our method, we provide with a thorough discussion on the finite sample properties. We suggest to get some preliminary information on the magnitude of the non-zero estimates and design the sample size accordingly before running the PCMLE. This way, satisfactory results can be achieved, with possibly the smallest amount of observations.

\section{Proofs}
\label{proofs}
\subsection{Proof of Theorem \ref{strongconsistence}}
We prove by contradiction. See a similar method to show strong consistency in \cite{HafnerPreminger09}. If $\hat{\boldsymbol{\theta}}_{\lambda_N}$ does not converge to $\boldsymbol{\theta}_0$ almost surely, there exists a $\eta>0$ such that the set $F=\{\omega:\limsup_{N\rightarrow \infty}\|\hat{\boldsymbol{\theta}}_{\lambda_N}(\omega)-\boldsymbol{\theta}_0\|\ge\eta\}$ has a positive probability. Since $\Lambda:=\Theta\cap\{\boldsymbol{\theta}: \|\boldsymbol{\theta}-\boldsymbol{\theta}_0\|\ge\eta\}$ is compact, for every $\omega\in F$, there exists a convergent subsequence $\{\hat{\boldsymbol{\theta}}_{\lambda_{N_i}}(\omega)\}$ such that
\begin{equation}
\{\hat{\boldsymbol{\theta}}_{\lambda_{N_i}}(\omega)\}\to \widetilde{\boldsymbol{\theta}}\in\Lambda. \nonumber
\end{equation}
It follows that
\begin{eqnarray}
  &&\limsup_{i\to\infty}\frac{1}{N_i}\left(\sum_{t=p+1}^{N_i}l_t(\boldsymbol{\theta}_0)-N_iP_{\lambda_{N_i}}
  (\boldsymbol{\theta}_0)\right) \label{1}\\
  &\leq& \limsup_{i\to\infty}\sup_{\boldsymbol{\theta}\in\Theta}\frac{1}{N_i}\left(\sum_{t=p+1}^{N_i}l_t(\boldsymbol{\theta})
  -N_iP_{\lambda_{N_i}}(\boldsymbol{\theta})\right) \nonumber\\
  &=& \limsup_{i\to\infty}\frac{1}{N_i}\left(\sum_{t=p+1}^{N_i}l_t(\hat{\boldsymbol{\theta}}_{\lambda_{N_i}}(\omega))
  -N_iP_{\lambda_{N_i}}(\hat{\boldsymbol{\theta}}_{\lambda_{N_i}}(\omega))\right) \nonumber\\
 &=& \limsup_{i\to\infty}\left(\frac{1}{N_i}\sum_{t=p+1}^{N_i}l_t(\hat{\boldsymbol{\theta}}_{\lambda_{N_i}}(\omega))
-P_{\lambda_{N_i}}(\tilde{\boldsymbol{\theta}})+P_{\lambda_{N_i}}(\tilde{\boldsymbol{\theta}})-P_{\lambda_{N_i}}(\hat{\boldsymbol{\theta}}_{\lambda_{N_i}}(\omega))\right) \nonumber\\
&=& \limsup_{i\to\infty}\frac{1}{N_i}\sum_{t=p+1}^{N_i}l_t(\hat{\boldsymbol{\theta}}_{\lambda_{N_i}}(\omega))\label{2}\\
&\leq& \limsup_{i\to\infty}\sup_{\boldsymbol{\theta}\in\Lambda}\frac{1}{N_i}\sum_{t=p+1}^{N_i}l_t(\boldsymbol{\theta})\le E\sup_{\boldsymbol{\theta}\in\Lambda}l_t(\boldsymbol{\theta})
.\label{3}
\end{eqnarray}
(\ref{2}) is from the conditions on the penalty function. We have (\ref{3}) from Lemma \ref{supl} since the first part of Assumption \ref{as1} implies $E\log^+|Z|<\infty$. On the other hand, 
\begin{eqnarray}
(\ref{1})=\lim_{N\to\infty}\frac{1}{N}\sum_{t=p+1}^{N}l_t(\boldsymbol{\theta}_0)-\lim_{N\to\infty}P_{\lambda_{N}}
  (\boldsymbol{\theta}_0)
=El_t(\boldsymbol{\theta}_0)
\end{eqnarray}
by the condition on the penalty function and Lemma \ref{lttheta}, part 1. 
Therefore, $El_t(\boldsymbol{\theta}_0)\le E\sup_{\boldsymbol{\theta}\in\Lambda}l_t(\boldsymbol{\theta})$ with a positive probability. But $\sup_{\boldsymbol{\theta}\in\Lambda}l_t(\boldsymbol{\theta})=l_t(\boldsymbol{\theta}_\Lambda)$ for some $\boldsymbol{\theta}_\Lambda\in \Lambda$ by the the continuity of $l_t(\cdot)$. This is a contradiction with Lemma \ref{lttheta}, part 2 since $||\boldsymbol{\theta}_\Lambda-\boldsymbol{\theta}||\ge\eta>0$.

\subsection{Proof of Proposition \ref{sparsity}}
\noindent We follow the pattern of the proof of Lemma 1 in \cite{FanLi01}. However, in our case, the estimation of the orders is completely different from theirs. We are considering a dependent case while theirs is for i.i.d. random variables.\\
\noindent To show $Q(\boldsymbol{0}^T, \boldsymbol{\theta}^T_{1,1})^T=\max_{||\boldsymbol{\theta}_{1,0}||\le CN^{-1/2}}Q(\boldsymbol{\theta})$, it is sufficient to have 
\begin{equation}\label{condition}
\frac{\partial Q(\boldsymbol{\theta})}{\partial \phi_j}<0 \;\;\;\;\text{for}\;\;\;0<\phi_j<CN^{-1/2}\;\;\text{and}\;\;
\frac{\partial Q(\boldsymbol{\theta})}{\partial \phi_j}>0 \;\;\;\;\text{for}\;\;\;-CN^{-1/2}<\phi_j<0
\end{equation}
for $1\le j \le s$.  By Taylor's expansion,
\begin{eqnarray}\label{Qdecomp}
&&\frac{\partial Q(\boldsymbol{\theta})}{\partial \phi_j}=\frac{\partial L(\boldsymbol{\theta})}{\partial \phi_j}-Np'_{\lambda_N}(|\phi_j|)sgn(\phi_j)\nonumber
=\frac{\partial L(\boldsymbol{\theta}_0)}{\partial \phi_j}+\sum_{i=1}^p\frac{\partial^2 L(\boldsymbol{\theta}_0)}{\partial \phi_j\partial \phi_i}(\phi_i-\phi_{i,0})\nonumber\\
&+&\frac{1}{2}\sum_{i=1}^p\sum_{k=1}^p\frac{\partial^3 L(\boldsymbol{\theta}^*)}{\partial \phi_j\partial \phi_i\partial \phi_k}(\phi_i-\phi_{i,0})(\phi_k-\phi_{k,0})-Np'_{\lambda_N}(|\phi_j|)sgn(\phi_j).
\end{eqnarray}
Here, $\boldsymbol{\theta}^*$ is between $\boldsymbol{\theta}$ and $\boldsymbol{\theta}_0$. By the observation (\ref{FtoG}), it is easy to see that
\begin{equation}
\begin{split}
\frac{\partial L(\boldsymbol{\theta})}{\partial \phi_j} &=\sum_{t=p+1}^N  \frac{\partial l_t(\boldsymbol{\theta})}{\partial \phi_j} 
=\sum_{t=p+1}^N  \frac{\partial \log g(X_t-\phi_1X_{t-1}-\cdots-\phi_p X_{t-p})}{\partial \phi_j}\nonumber\\
&=-\sum_{t=p+1}^N  \frac{ g'(X_t-\phi_1X_{t-1}-\cdots-\phi_p X_{t-p})}{ g(X_t-\phi_1X_{t-1}-\cdots-\phi_p X_{t-p})}X_{t-j},\nonumber
\end{split}
\end{equation}
\begin{eqnarray}
\frac{\partial^2 L(\boldsymbol{\theta})}{\partial \phi_j\partial \phi_i}
&=&\sum_{t=p+1}^N \left( \frac{ g''g-(g')^2}{ g^2}\right)(X_t-\phi_1X_{t-1}-\cdots-\phi_p X_{t-p})X_{t-j}X_{t-i},\nonumber
\end{eqnarray}
\begin{eqnarray}
\frac{\partial^3 L(\boldsymbol{\theta})}{\partial \phi_j\partial \phi_i\partial \phi_k}
&=&-\sum_{t=p+1}^N  \left(\frac{g'}{g}\right)''(X_t-\phi_1X_{t-1}-\cdots-\phi_p X_{t-p})X_{t-j}X_{t-i}X_{t-k}.\nonumber
\end{eqnarray}
Therefore, the first term of (\ref{Qdecomp}) is 
\begin{equation}\label{Firstderivative}
\frac{\partial L(\boldsymbol{\theta}_0)}{\partial \phi_j}=-\sum_{t=p+1}^N  \frac{ g'(Z_t)}{ g(Z_t)}X_{t-j}.
\end{equation}
In the second term of (\ref{Qdecomp}),  
\begin{eqnarray}
\frac{\partial^2 L(\boldsymbol{\theta}_0)}{\partial \phi_j\partial \phi_i}
&=&\sum_{t=p+1}^N  \frac{ g''(Z_t)g(Z_t)-(g'(Z_t))^2}{ g^2(Z_t)}X_{t-j}X_{t-i}.\nonumber
\end{eqnarray}

\noindent First, we estimate the order of (\ref{Firstderivative}), the first term of (\ref{Qdecomp}).  Since
\begin{equation}
E\frac{ g'(Z_t)}{ g(Z_t)}=\int g'(z)dz=0,
\end{equation}
\begin{equation}
E\left(\frac{ g'(Z_t)}{ g(Z_t)}\right)^2X_{t-j}^2=E\left(\frac{ g'(Z_t)}{ g(Z_t)}\right)^2EX_{t-j}^2=C(g)\gamma(0)<\infty,
\end{equation}
and for $s< t$
\begin{equation}
E\frac{ g'(Z_s)}{ g(Z_s)}X_{s-j}\frac{ g'(Z_t)}{ g(Z_t)}X_{t-j}=E\frac{ g'(Z_t)}{ g(Z_t)}EX_{t-j}X_{s-j}\frac{ g'(Z_s)}{ g(Z_s)}=0,
\end{equation}
we have
\begin{eqnarray}
E\left(\sum_{t=p+1}^N  \frac{ g'(Z_t)}{ g(Z_t)}X_{t-j}\right)^2&=& \sum_{t=p+1}^N  E\left(\frac{ g'(Z_t)}{ g(Z_t)}X_{t-j}\right)^2\nonumber\\
&=& (N-p)C(g)\gamma(0).\nonumber
\end{eqnarray}
Therefore $||\frac{\partial L(\boldsymbol{\theta}_0)}{\partial \phi_j}||=O(N^{1/2})$. By Chebyshev's inequality, we have
\begin{equation}\label{1derivative}
\frac{\partial L(\boldsymbol{\theta}_0)}{\partial \phi_j}=O_P(N^{1/2}).
\end{equation}
Next we estimate the order of the second term of (\ref{Qdecomp}).
\begin{eqnarray}
E\frac{ g''(Z_t)g(Z_t)-(g'(Z_t))^2}{ g^2(Z_t)}=\int g''(z)dz-C(g)=-C(g)\nonumber
\end{eqnarray}
by the assumptions on $g$. Denote $Y_t=\left(\frac{ g''(Z_t)g(Z_t)-(g'(Z_t))^2}{ g^2(Z_t)}+C(g)\right)X_{t-j}X_{t-i}$. Then
$EY_t=EY_tY_s=0$ for $s\ne t$, and
\begin{eqnarray}
EY_t^2&=&E\left(\frac{ g''(Z_t)g(Z_t)-(g'(Z_t))^2}{ g^2(Z_t)}+C(g)\right)^2EX_{t-j}^2X_{t-i}^2\nonumber\\
&=&\left(E\frac{ (g''(Z_t))^2g^2(Z_t)-2g''(Z_t)(g'(Z_t))^2g(Z_t)+(g'(Z_t))^4}{ g^4(Z_t)}-C^2(g)\right)EX_{t-j}^2X_{t-i}^2\nonumber\\
&=&\left(E\frac{ (g''(Z_t))^2}{ g^2(Z_t)}-2E\frac{g''(Z_t)(g'(Z_t))^2}{g^3(Z_t)}+E\frac{(g'(Z_t))^4}{ g^4(Z_t)}-C^2(g)\right)EX_{t-j}^2X_{t-i}^2\nonumber
\end{eqnarray}
is finite by the assumptions on $g$ and Lemma \ref{lemma3}. Notice that $E\frac{g''(Z_t)(g'(Z_t))^2}{g^3(Z_t)}<\infty$ by Cauchy–-Schwarz inequality. Then by the weak law of large numbers (Theorem 8.3.2, \cite{Dudley02}), $\sum_{t=p+1}^NY_t/N$ converges to $0$ in probability. Besides, $\sum_{t=p+1}^NX_{t-j}X_{t-i}/N$ converges to $\gamma(j-i)$ in probability by the ergodicity of (\ref{arp}). See the ergodicity of (\ref{arp}) in the proof of Lemma \ref{supl}. Therefore,
\begin{equation}\label{2derivative}
\frac{\partial^2 L(\boldsymbol{\theta}_0)}{\partial \phi_j\partial \phi_i}/N \;\;\;\;\text{converges to}\;\;\; -C(g)\gamma(j-i) \;\;\;\text{in probability.}
\end{equation}
Hence, the second term of (\ref{Qdecomp}) has order $O_P(N^{1/2})$. Besides,  $\frac{\partial^3 L(\boldsymbol{\theta}^*)}{\partial \phi_j\partial \phi_i\partial \phi_k}=o_P(N^{3/2})$ by Lemma \ref{3product}. Therefore, the first three terms of (\ref{Qdecomp}) have order $O_P(N^{1/2})$ by the condition on $\boldsymbol{\theta}$. By the conditions on $\lambda_N$ and $p'_{\lambda_N}(\phi_j)$, the last term $Np'_{\lambda_N}(|\phi_j|)sgn(\phi_j)$ is dominating the other three terms in (\ref{Qdecomp}).
 (\ref{condition}) is established. This completes the proof.

\subsection{Proof of Proposition \ref{weakconsistence}}

This theorem is a version of Theorem 1 in \cite{FanLi01} for dependent random variables. We sketch the proof for completeness. Let $b_N=N^{-1/2}+a_N$ and  $\boldsymbol{u}=(u_1,\cdots,u_p)^T$. To show the existence of a local maximizer with $||\hat{\boldsymbol{\theta}}-\boldsymbol{\theta}_0||=O_P(b_N)$, for any $\eta>0$, it is sufficient to have $P(\sup_{||\boldsymbol{u}||=C}Q(\boldsymbol{\theta}_0+b_N\boldsymbol{u})\le  Q(\boldsymbol{\theta}_0))\ge 1-\eta$ for some large constant $C$ .  In our case,
\begin{eqnarray}
D_N(\boldsymbol{u})&:=&Q(\boldsymbol{\theta}_0+b_N\boldsymbol{u})-Q(\boldsymbol{\theta}_0)\notag\\
&\le& L(\boldsymbol{\theta}_0+b_N\boldsymbol{u})-L(\boldsymbol{\theta}_0)-N\sum_{j=s+1}^p\{p_{\lambda_N}(|\phi_{j,0}+b_N u_j|)-p_{\lambda_N}(|\phi_{j,0}|)\}\notag\\
&\le&b_NL'(\boldsymbol{\theta}_0)^T\boldsymbol{u}+\frac{1}{2}\boldsymbol{u}^TH(\boldsymbol{\theta}_0)\boldsymbol{u}b_N^2+|\sum_{i=1}^p\sum_{j=1}^p\sum_{k=1}^p\frac{\partial^3L(\boldsymbol{\theta}^*)}{\partial\phi_i\partial\phi_j\partial\phi_k}b_N^3u_iu_ju_k|\notag\\
&-&N\sum_{j=s+1}^p b_Np_{\lambda_N}'(|\phi_{j,0}|)sgn(\phi_{j,0})u_j-\frac{N}{2}\sum_{j=s+1}^p b_N^2p_{\lambda_N}''(|\phi_{j,0}|)u_j^2\{1+o(1)\}\notag
\end{eqnarray}
Here the gradient $L'(\boldsymbol{\theta}_0)=(-\sum_{t=p+1}^T \frac{ g'(Z_t)}{ g(Z_t)}X_{t-1}, \cdots, -\sum_{t=p+1}^N  \frac{ g'(Z_t)}{ g(Z_t)}X_{t-p})^T$, and the matrix $H(\boldsymbol{\theta}_0)=(a_{ij})_{p\times p}$ with $a_{ij}=\frac{\partial^2 L(\boldsymbol{\theta}_0)}{\partial \phi_i\partial \phi_j}
=\sum_{t=p+1}^N  \frac{ g''(Z_t)g(Z_t)-(g'(Z_t))^2}{ g^2(Z_t)}X_{t-i}X_{t-j}$.
By (\ref{1derivative}), we have $\frac{\partial L(\boldsymbol{\theta}_0)}{\partial \phi_j}=O_P(N^{1/2})$. Then the first term has order $O_P(N^{1/2}b_N)=O_P(Nb_N^2)$.  Notice that $a_{ij}=O_P(N)$ by (\ref{2derivative}). Therefore, the second term has order $O_P(Nb_N^2)$. The third term has order $O_p(Nb_N^3)=o_p(Nb_N^2)$ by Lemma \ref{3product} and the condition on $a_N$ in Assumptions \ref{as3}. Recall that $b_N=N^{-1/2}+a_N$. It is obvious that the fourth term has order $O_P(Nb_N^2)$ and the fifth term has order $o_P(Nb_N^2)$. Then the second term dominates the others by choosing a sufficiently large $C$.  Let $\Sigma$ be the non-negative definite  ${p\times p}$ matrix with the entry $\gamma(j-i)$ at row $j$ and column $i$ , $1\le i, j\le p$. Again by (\ref{2derivative}), $H(\boldsymbol{\theta}_0)/N$ converges to $-C(g)\Sigma$ with $C(g)\ge 0$. Therefore the second term is  non-positive with probability tending to 1. This finishes the proof of the proposition.
\subsection{Proof of Theorem \ref{thm2}}
We only need to show the second part. 
Let $\hat{\boldsymbol{\theta}}=(\textbf{0},\hat{\boldsymbol{\theta}}_{1,1})$ be the $N^{-1/2}$ consistent local maximizer  of $Q(\boldsymbol{\theta})$ with  $\frac{\partial Q(\hat{\boldsymbol{\theta}})}{\partial\phi_j}=0 \;\;\;\text{for}\;\;\;j=s+1,\cdots,p$.
Then we have 
\begin{eqnarray}\label{talor0}
\frac{\partial Q(\hat{\boldsymbol{\theta}})}{\partial\phi_j}&=&\frac{\partial L(\hat{\boldsymbol{\theta}})}{\partial\phi_j}-Np'_{\lambda_N}(|\hat\phi_j|)sgn(\hat\phi_j)\notag\\
&=&\frac{\partial L(\boldsymbol{\theta}_0)}{\partial\phi_j}+\sum_{i=s+1}^p[\frac{\partial^2L(\boldsymbol{\theta}_0)}{\partial\phi_j\phi_i}+o_P(N)](\hat\phi_i-\phi_{i,0})\notag\\
&-&Np'_{\lambda_N}(|\phi_{j,0}|)sgn(\phi_{j,0})-N[p_{\lambda_N}''(|\phi_{j,0}|)+o_P(1)](\hat\phi_j-\phi_{j,0})\notag\\
&=&0.
\end{eqnarray}
Here the $o_p(N)$ in the second term is from Lemma \ref{3product} and the $N^{-1/2}$ consistency of $\hat{\boldsymbol{\theta}}_{1,1}$. The $o_p(1)$ in the last term is from the property of $p'''_{\lambda_N}$. Let
$$M'(\boldsymbol{\theta}_{0,1})=\left(\frac{\partial L(\boldsymbol{\theta}_0)}{\partial \phi_{s+1}},\cdots, \frac{\partial L(\boldsymbol{\theta}_0)}{\partial \phi_{p}}\right)^T$$
be the gradient vector and $M''(\boldsymbol{\theta}_{0,1})$ be the second partial derivative Hessian matrix of $L(\boldsymbol{\theta})$ at $\boldsymbol{\theta}_{0,1}$. Then the matrix form of (\ref{talor0}) is
\begin{eqnarray}
M'(\boldsymbol{\theta}_{0,1})+M''(\boldsymbol{\theta}_{0,1})(\hat{\boldsymbol{\theta}}_{1,1}-\boldsymbol{\theta}_{0,1})-N\boldsymbol{b}-N(\Delta+o_P(\textbf{1}))(\hat{\boldsymbol{\theta}}_{1,1}-\boldsymbol{\theta}_{0,1})=\textbf{0}.\notag
\end{eqnarray}
Divided by $\sqrt{N}$, together with some algebra, we have
\begin{eqnarray}
&&\frac{1}{\sqrt N}M'(\boldsymbol{\theta}_{0,1})+\sqrt{N}(\frac{1}{N}M''(\boldsymbol{\theta}_{0,1})+C(g)\Gamma)(\hat{\boldsymbol{\theta}}_{1,1}-\boldsymbol{\theta}_{0,1})\notag\\
&-&\sqrt{N}(C(g)\Gamma+\Delta+o_P(\textbf{1}))(\hat{\boldsymbol{\theta}}_{1,1}-\boldsymbol{\theta}_{0,1})-\sqrt{N}\boldsymbol{b}=\textbf{0}.\notag
\end{eqnarray}
By (\ref{2derivative}), 
$$\frac{1}{N}M''(\boldsymbol{\theta}_{0,1})+C(g)\Gamma\rightarrow \textbf{0}\;\;\text{ in probability.}$$ 
Therefore, to have the second part of Theorem \ref{thm2}, we only need to prove $\frac{1}{\sqrt N}M'(\boldsymbol{\theta}_{0,1})\Rightarrow N(0, C(g)\Gamma)$. By Cram\'{e}r-Wold device, it is enough to have 
\begin{equation}\label{linearCLT}
\frac{1}{\sqrt N}\sum_{j=s+1}^p\lambda_j\frac{\partial L(\boldsymbol{\theta}_0)}{\partial \phi_j}\Rightarrow N(0, C(g)\sum_{s+1\le i,j\le p}\lambda_i\lambda_j\gamma(i-j))
\end{equation}
for any vector $\boldsymbol{\lambda}=(\lambda_{s+1}, \cdots, \lambda_p)^T$ with $||\boldsymbol{\lambda}||\ne 0$. Let $\mathcal{F}_t$ be the $\sigma$-algebra $\sigma(X_1,\cdots, X_t)$ and  $E_t(\cdot):=E(\cdot|\mathcal{F}_t)$ be the  conditional expectation.
$$\sum_{j=s+1}^p\lambda_j\frac{\partial L(\boldsymbol{\theta}_0)}{\partial \phi_j}=-\sum_{t=p+1}^N  \frac{ g'(Z_t)}{ g(Z_t)}\sum_{j=s+1}^p\lambda_j X_{t-j}$$ is a Martingale since
\begin{equation}
E_{t-1}\left(\frac{ g'(Z_t)}{ g(Z_t)}\sum_{j=s+1}^p\lambda_j X_{t-j}\right)=\left(\sum_{j=s+1}^p\lambda_j X_{t-j}\right)E\left(\frac{ g'(Z_t)}{ g(Z_t)}\right)=0.\notag
\end{equation}
Now we can use the Lindeberg condition given in \cite{Brown71} for the Martingale central limit theorem. First, we verify the condition (1) on page 60 in his paper. 
\begin{eqnarray}
\sigma_t^2&:=&E_{t-1}\left( \left(\frac{ g'(Z_t)}{ g(Z_t)}\sum_{j=s+1}^p\lambda_j X_{t-j}\right)^2\right)\notag\\
&=&\left(\sum_{j=s+1}^p\lambda_j X_{t-j}\right)^2E\left(\frac{ g'(Z_t)}{ g(Z_t)}\right)^2=C(g)\left(\sum_{j=s+1}^p\lambda_j X_{t-j}\right)^2,\notag
\end{eqnarray}
$V_N^2:=\sum_{t=p+1}^N\sigma_t^2$ and 
\begin{equation}\label{variance}
\begin{split}
s_N^2:&=EV_N^2=C(g)\sum_{t=p+1}^NE\left(\sum_{j=s+1}^p\lambda_j X_{t-j}\right)^2\\
&=(N-p)C(g)\sum_{s+1\le i,j\le p}\lambda_i\lambda_j\gamma(i-j)=O(N).
\end{split}
\end{equation}
Then
\begin{eqnarray}
V_N^2/s_N^2&=&\frac{\sum_{t=p+1}^N\left(\sum_{j=s+1}^p\lambda_j X_{t-j}\right)^2}{\sum_{t=p+1}^NE\left(\sum_{j=s+1}^p\lambda_j X_{t-j}\right)^2}\notag\\
&=&\frac{\sum_{j=s+1}^p\lambda_j^2\sum_{t=p+1}^N X_{t-j}^2+2\sum_{j=s+1}^{p-1}\sum_{k=j+1}^p\lambda_j\lambda_k\sum_{t=p+1}^NX_{t-j}X_{t-k}}{\sum_{j=s+1}^p\lambda_j^2E\sum_{t=p+1}^N X_{t-j}^2+2\sum_{j=s+1}^{p-1}\sum_{k=j+1}^p\lambda_j\lambda_kE\sum_{t=p+1}^NX_{t-j}X_{t-k}}.\notag
\end{eqnarray}
To show $V_N^2/s_N^2\rightarrow 1$ in probability, it is sufficient to have 
\begin{equation}\label{square1}
\frac{\sum_{t=p+1}^N X_{t-j}^2}{E\sum_{t=p+1}^N X_{t-j}^2}\rightarrow 1\;\;\;\text{in probability}
\end{equation}
and 
\begin{equation}\label{alternate1}
\frac{\sum_{t=p+1}^NX_{t-j}X_{t-k}}{E\sum_{t=p+1}^NX_{t-j}X_{t-k}}\rightarrow 1\;\;\;\text{in probability}
\end{equation}
for any $s+1\le j<k\le p$. (\ref{square1}) and (\ref{alternate1}) are true from the ergodicity of (\ref{arp}).  See the ergodicity of (\ref{arp}) in the proof of Lemma \ref{supl}.  Now we verify the Lindeberg condition as in \cite{Brown71}. Let $P_t(\cdot)=P(\cdot|\mathcal{F}_t)$ be the conditional probability. For any $\epsilon>0$, 
\begin{eqnarray}
&&E_{t-1}\left(\left(\frac{ g'(Z_t)}{ g(Z_t)}\sum_{j=s+1}^p\lambda_j X_{t-j}\right)^2 \textbf{1}(|\frac{ g'(Z_t)}{ g(Z_t)}\sum_{j=s+1}^p\lambda_j X_{t-j}|\ge \epsilon s_N)\right)\notag\\
&=&\left(\sum_{j=s+1}^p\lambda_j X_{t-j}\right)^2 E_{t-1}\left(\left(\frac{ g'(Z_t)}{ g(Z_t)}\right)^2 \textbf{1}(|\frac{ g'(Z_t)}{ g(Z_t)}|\ge \epsilon s_N/\sum_{j=s+1}^p\lambda_j X_{t-j})\right)\notag\\
&\le&\left(\sum_{j=s+1}^p\lambda_j X_{t-j}\right)^2 \left(E\left(\frac{ g'(Z_t)}{ g(Z_t)}\right)^4\right)^{1/2} \left[P_{t-1}\left(|\frac{ g'(Z_t)}{ g(Z_t)}|\ge \epsilon s_N/\sum_{j=s+1}^p\lambda_j X_{t-j}|\right)\right]^{1/2}\notag\\
&\le&\left(\sum_{j=s+1}^p\lambda_j X_{t-j}\right)^2 \left(E\left(\frac{ g'(Z_t)}{ g(Z_t)}\right)^4\right)^{1/2} \left[\frac{(\sum_{j=s+1}^p\lambda_j X_{t-j})^2E\left(\frac{ g'(Z_t)}{ g(Z_t)}\right)^2}{ \epsilon^2 s_N^2}\right]^{1/2}\notag\\
&=&\left(\sum_{j=s+1}^p\lambda_j X_{t-j}\right)^3 \left(E\left(\frac{ g'(Z_t)}{ g(Z_t)}\right)^4\right)^{1/2} \left(E\left(\frac{ g'(Z_t)}{ g(Z_t)}\right)^2 \right)^{1/2}\epsilon^{-1} s_N^{-1}.\notag
\end{eqnarray}
By a similar argument as in Lemma \ref{3product}, $\sum_{j=p+1}^N\left(\sum_{j=s+1}^p\lambda_j X_{t-j}\right)^3=o_P(N^{3/2})$. Therefore,
\begin{eqnarray}
&&s_N^{-2}\sum_{j=p+1}^NE_{t-1}\left(\left(\frac{ g'(Z_t)}{ g(Z_t)}\sum_{j=s+1}^p\lambda_j X_{t-j}\right)^2 \textbf{1}(|\frac{ g'(Z_t)}{ g(Z_t)}\sum_{j=s+1}^p\lambda_j X_{t-j}|\ge \epsilon s_N)\right)\notag\\
&=&o_P(1).\notag
\end{eqnarray}
The Lindeberg condition is satisfied. Finally, the variance in the central limit theorem (\ref{linearCLT}) is from the calculation (\ref{variance}). This finishes the proof.


\section{Appendix}

\begin{lemma}\label{supl}
Assume the AR model (\ref{arp}) is causal, $E\log^+|Z|:=E\{\max(0, \log|Z|)\}<\infty$ and  $g(z)$ is continuous.  Let $\Lambda$ be a compact subset of the parameter space $\Theta$. Then $\left\{\sup_{\boldsymbol{\theta}\in\Lambda}l_t(\boldsymbol{\theta})\right\}$ is strictly stationary, ergodic and 
$$\limsup_{N\rightarrow \infty}\sup_{\boldsymbol{\theta}\in\Lambda}\frac{1}{N}\sum_{i=p+1}^N l_i(\boldsymbol{\theta})\le E\sup_{\boldsymbol{\theta}\in\Lambda}l_t(\boldsymbol{\theta})\;\; a.s.$$
\end{lemma}
\begin{proof}
Let $\textbf{X}:=(X_t,X_{t-1},\cdots,X_{t-p})^T$ be the vector of $p+1$ random variables. Recall (\ref{FtoG}). To emphasize the dependence of $l_t(\boldsymbol{\theta})$ on $\textbf{X}$, denote $l_t(\textbf{X},\boldsymbol{\theta}):=l_t(\boldsymbol{\theta})=\log g(X_t-\sum_{j=1}^{p}\phi_jX_{t-j})$. 
Then $l_t(\textbf{X},\boldsymbol{\theta})$ is continuous by the continuity of $g(z)$. We claim that $\sup_{\boldsymbol{\theta}\in\Pi}l_t(\textbf{X},\boldsymbol{\theta})$  is continuous with respect to $\textbf{X}$, for any compact subset $\Pi\subset\Lambda$.  Assume by contradiction that
 $\sup_{\boldsymbol{\theta}\in\Pi}l_t(\textbf{X},\boldsymbol{\theta})$  is  not continuous at $\textbf{X}_{(0)}$. Then there exists an $\epsilon >0$ such that for all $\delta >0$, there exists a $\textbf{X}_{(1)}$,
 $$\left\|\textbf{X}_{(1)}-\textbf{X}_{(0)}\right\|<\delta \;\;\text{ and} \;\;|\sup_{\boldsymbol{\theta}\in\Pi}l_t(\textbf{X}_{(1)},\boldsymbol{\theta})-\sup_{\boldsymbol{\theta}\in\Pi}l_t(\textbf{X}_{(0)},\boldsymbol{\theta})|>\epsilon.$$   
By the continuity of $l_t(\textbf{X},\boldsymbol{\theta})$ with respect to $\boldsymbol{\theta}=(\phi_1,\cdots,\phi_p)^T$, 
$\sup_{\boldsymbol{\theta}\in\Pi}l_t(\textbf{X},\boldsymbol{\theta})$ 
is attained in $\Pi$ for each $\textbf{X}$ and each compact subset $\Pi$ of $\Lambda$. Denote $$l_t(\textbf{X}_{(0)},\boldsymbol{\theta}_{(0)})=\sup_{\boldsymbol{\theta}\in\Pi}l_t(\textbf{X}_{(0)},\boldsymbol{\theta})\;\;\text{ and}  \;\;l_t(\textbf{X}_{(1)},\boldsymbol{\theta}_{(1)})=\sup_{\boldsymbol{\theta}\in\Pi}l_t(\textbf{X}_{(1)},\boldsymbol{\theta}).$$
Without loss of generality, assume $l_t(\textbf{X}_{(1)},\boldsymbol{\theta}_{(1)})>l_t(\textbf{X}_{(0)},\boldsymbol{\theta}_{(0)})$. Then 
$$l_t(\textbf{X}_{(1)},\boldsymbol{\theta}_{(1)})-l_t(\textbf{X}_{(0)},\boldsymbol{\theta}_{(1)})>l_t(\textbf{X}_{(1)},\boldsymbol{\theta}_{(1)})-l_t(\textbf{X}_{(0)},\boldsymbol{\theta}_{(0)})>\epsilon.$$ 
This is a contradiction with the continuity of $l_t(\textbf{X},\boldsymbol{\theta})$ with respect to $\textbf{X}$. 
Hence  $\sup_{\boldsymbol{\theta}\in\Pi}l_t(\textbf{X},\boldsymbol{\theta})$ is continuous with respect to $\textbf{X}$, for any compact subset $\Pi\subset\Lambda$.  
Consequently  $\sup_{\boldsymbol{\theta}\in\Pi}l_t(\textbf{X},\boldsymbol{\theta})$ is $\mathcal{B}$ measurable, where $\mathcal{B}$ is the Borel $\sigma$-algebra on $\mathbb{R}^{p+1}$. This verifies the second condition of Theorem 3.10 in \cite{Pfanzagl69}. The other two conditions are obvious by the continuity of $l_t(\textbf{X},\boldsymbol{\theta})$. Besides, $\Lambda$ is a compact set in $\Theta$. Therefore, by Theorem 3.10 of \cite{Pfanzagl69}, there exists a $\mathcal{B}$ measurable function $\boldsymbol{\varphi}(\textbf{X})=(\varphi_1(\textbf{X}),\cdots,\varphi_p(\textbf{X}))^T: \mathbb{R}^{p+1}\to \Lambda$ such that
\begin{eqnarray}\label{supexp}
\sup_{\boldsymbol{\theta}\in\Lambda}l_t(\textbf{X},\boldsymbol{\theta})=\log g(X_t-\sum_{j=1}^p\varphi_j(\textbf{X})X_{t-j})
\end{eqnarray}
 Since the AR(p) model (\ref{arp}) is causal and $E\log^+|Z|<\infty$,  (\ref{arp}) is strictly stationary by Theorem 1 in \cite{BrockwellLindner10}. On the other hand, (\ref{arp}) is ergodic: $Z_t$ has a continuous density function $g(z)$ implies that its law is absolutely continuous with respect to the Lebesgue measure on $\mathbb{R}$. Therefore, (\ref{arp}) is strong mixing (\cite{Mokkadem88}) and then  is ergodic (problem 24.3, \cite{Billingsley95}). By (\ref{supexp}) and the continuity of $g(z)$, the time series $\left\{\sup_{\boldsymbol{\theta}\in\Lambda}l_t(\boldsymbol{\theta})\right\}$ is strictly stationary and ergodic (Theorem 36.4, \cite{Billingsley95}). Therefore,
\begin{eqnarray}
\limsup_{N\rightarrow \infty}\sup_{\boldsymbol{\theta}\in\Lambda}\frac{1}{N}\sum_{i=p+1}^N l_i(\boldsymbol{\theta})
\le\limsup_{N\rightarrow \infty}\frac{1}{N}\sum_{i=p+1}^N \sup_{\boldsymbol{\theta}\in\Lambda}l_i(\boldsymbol{\theta})
= E\sup_{\boldsymbol{\theta}\in\Lambda}l_t(\boldsymbol{\theta})\;\; a.s.\notag
\end{eqnarray}

\end{proof}

\begin{lemma}\label{lttheta}
 Assume  $Z, Z_t$ are i.i.d. and  $E|\log g(Z)|<\infty$, we have
\begin{enumerate}
 \item $\lim_{N\rightarrow\infty}\frac{1}{N}\sum_{t=p+1}^Nl_t(\boldsymbol{\theta}_0)=E\log g(Z)=El_t(\boldsymbol{\theta}_0);$
 \item $El_t(\boldsymbol{\theta}_0)\ge El_t(\boldsymbol{\theta})\; \text{with equality if and only if}\;\; \boldsymbol{\theta}=\boldsymbol{\theta}_0.$
\end{enumerate}
\end{lemma}
\begin{proof}
1. Recall $l_t(\boldsymbol{\theta}_0)=\log g(Z_t)$ from (\ref{Ltheta0}). Then $|El_t(\boldsymbol{\theta}_0)|=|E\log g(Z)|<\infty$. 
By the law of large number\textbf{s},
\begin{eqnarray}
&&\lim_{N\rightarrow\infty}\frac{1}{N}\sum_{t=p+1}^Nl_t(\boldsymbol{\theta}_0)\nonumber\\
&=&\lim_{N\rightarrow\infty}\frac{1}{N}\sum_{t=p+1}^N\log g(Z_t)=E\log g(Z)=El_t(\boldsymbol{\theta}_0) \;\;a.s.\nonumber
\end{eqnarray}
2. First, we show $E\log g(Z+C)\le E\log g(Z)$ for any constant $C$ and the equality holds if and only if $C=0$. Recall $g$ is the density function of $Z$. Obviously, the equality holds if $C=0$. If $C\ne 0$, by the strict concavity of  the logarithm function,
\begin{equation}
\label{contrast}
\begin{split}
E\log g(Z+C)-E\log g(Z) &=E\log\frac{g(Z+C)}{g(Z)}<\log E\frac{g(Z+C)}{g(Z)}\\
                        &=\log\int g(z+C)dz=0.
\end{split}
\end{equation}
(\ref{contrast}) is a simplified version of Example (1.3) in \cite{Pfanzagl69}. We provide the proof here for completeness. Now let $X=(\phi_{1,0}-\phi_1)X_{t-1}+\cdots+(\phi_{p,0}-\phi_p)X_{t-p}$. Since $Z_t$ is independent of $X$,  by (\ref{contrast}), we have
\begin{eqnarray}
El_t(\boldsymbol{\theta})&=&E\log g(X_t-\phi_1X_{t-1}-\cdots-\phi_p X_{t-p})\nonumber\\
&=&E\log g(X_t-\phi_{1,0}X_{t-1}-\cdots-\phi_{p,0}X_{t-p}+(\phi_{1,0}-\phi_1)X_{t-1}+\cdots+(\phi_{p,0}-\phi_p)X_{t-p})\nonumber\\
&=&E\log g(Z_t+X)= E\{E(\log g(Z_t+X)|X)\}\nonumber\\
&\le& E(E \log g(Z_t))=E\log g(Z_t)=El_t(\boldsymbol{\theta}_0).\nonumber
\end{eqnarray}
This completes the proof.
\end{proof}
\begin{remark}
In the case $Z, Z_t\sim N(0,1)$,
\begin{eqnarray}
El_t(\boldsymbol{\theta})&=&E\log g(Z_t+X)=-\frac{1}{2}\log 2\pi-\frac{1}{2}E (Z_t+X)^2\nonumber\\
&=&-\frac{1}{2}\log 2\pi-\frac{1}{2}EZ_t^2-\frac{1}{2}EX^2\nonumber\\
&=&E\log g(Z)-\frac{1}{2}E X^2=El_t(\boldsymbol{\theta}_0)-\frac{1}{2}E X^2.\nonumber
\end{eqnarray}
Obviously, the second part of Lemma \ref{lttheta} is true in this case.
\end{remark}

\begin{lemma}\label{3product}
Assume Assumptions \ref{as1} and part 3 of Assumptions \ref{as2}. Further, assume $Z$ has the first three moments. Then $\sum_{t=1}^N \left(\frac{g'}{g}\right)''(X_t-\phi_1X_{t-1}-\cdots-\phi_p X_{t-p})X_tX_{t-i}X_{t-k}=O_P(N)$, for any given integers $i,\; k$.
\end{lemma}
\begin{proof}
Let $A=\max(E|Z_1|^3, EZ_1^2E|Z_2|, (E|Z_1|)^3)$. 
Under the causality condition, $X_t=\sum_{j=0}^\infty a_jZ_{t-j}$ for a sequence of constants  $a_j$ with $\sum_{j=0}^\infty |a_j|<\infty$. 
\begin{equation}\label{e3}
\begin{split}
  &E|X_tX_{t-i}X_{t-k}|=E|\sum_{j=0}^{\infty}\sum_{p=0}^{\infty}\sum_{q=0}^{\infty}a_ja_pa_qZ_{t-j}Z_{t-i-p}Z_{t-k-q}|\\
  &\le\sum_{j=0}^{\infty}\sum_{p=0}^{\infty}\sum_{q=0}^{\infty}E|a_ja_pa_qZ_{t-j}Z_{t-i-p}Z_{t-k-q}|\\
  &\le A\sum_{j=0}^{\infty}\sum_{p=0}^{\infty}\sum_{q=0}^{\infty}|a_ja_pa_q|=A(\sum_{j=0}^{\infty}|a_j|)^3<\infty.
\end{split}
\end{equation}
Therefore,
\begin{equation*}
\begin{split}
&E|\sum_{t=1}^N \left(\frac{g'}{g}\right)''X_tX_{t-i}X_{t-k}|\le \sum_{t=1}^N E|\left(\frac{g'}{g}\right)''X_tX_{t-i}X_{t-k}| \\
&\le B \sum_{t=1}^N E|X_tX_{t-i}X_{t-k}|=O(N).
\end{split}
\end{equation*}
Then by Markov's inequality, the desired result follows.
\end{proof}

\begin{lemma}\label{lemma3}
Assume that the innovations $\{Z_t\}$ are i.i.d. random variables.  Under the causality condition, $X_t$ in the AR(p) model (\ref{arp}), has the $m^{th}$ moment if the corresponding innovation $Z_t$ has the $m^{th}$ moment. When the innovation has the fourth moment, $EX_t^2X_{t+k}^2<\infty$ for any given integer $k$.
\end{lemma}
\begin{proof}
Under the causality condition, $X_t=\sum_{j=0}^\infty a_jZ_{t-j}$ for a sequence of constants $a_j$ with $\sum_{j=0}^\infty |a_j|<\infty$. Without loss of generality, we assume that 
$\sum_{j=0}^\infty |a_j|=1$, and $a_j\ne 0$ for $j=1,2,\cdots$. By the convexity of the function $|x|^m$,
\begin{equation}
\begin{split}
  &\left|\sum_{j=0}^{\infty}a_jZ_{t-j}\right|^m\le\left(\sum_{j=0}^{\infty}|a_jZ_{t-j}|\right)^m =\left(|a_0||Z_t|+\sum_{j=1}^{\infty}|a_j||Z_{t-j}|\right)^m\\
  &\leq |a_0|\left|Z_t\right|^m+\left(\sum_{j=1}^{\infty}|a_j|\right)\left(\frac{|a_1||Z_{t-1}|+\sum_{j=2}^{\infty}|a_j||Z_{t-j}|}
  {\sum_{j=1}^{\infty}|a_j|}\right)^m\\
  &\leq |a_0|\left|Z_t\right|^m+\left(\sum_{j=1}^{\infty}|a_j|\right)\left(\frac{|a_1|}{\sum_{j=1}^{\infty}|a_j|}
  \left|Z_{t-1}\right|^m+\frac{\sum_{j=2}^{\infty}|a_j|}{\sum_{j=1}^{\infty}|a_j|}\left(
  \frac{\sum_{j=2}^{\infty}|a_j||Z_{t-j}|}{\sum_{j=2}^\infty|a_j|}\right)^m\right)\\
  &=|a_0|\left|Z_t\right|^m+|a_1|\left|Z_{t-1}\right|^m+\left(\sum_{j=2}^{\infty}|a_j|\right)\left(\frac{|a_2||Z_{t-2}|+\sum_{j=3}^{\infty}|a_j||Z_{t-j}|}
  {\sum_{j=2}^{\infty}|a_j|}\right)^m\\
  &\leq\cdots\leq\sum_{j=0}^{\infty}|a_j|\left|Z_{t-j}\right|^m. \nonumber  
\end{split}
\end{equation}
Taking expectations on both sides, we obtain:
\begin{equation}
  E\left|X_t\right|^m\leq E|Z|^m<\infty\nonumber
\end{equation}
for any positive integer m. Now,
\begin{equation}
\begin{split}
  EX_t^2X_{t-k}^2
  &\leq E\left(\frac{X_t+X_{t-k}}{2}\right)^4
  =\frac{1}{16}E\left(\sum_{j=0}^{\infty}a_jZ_{t-j}+\sum_{j=0}^{\infty}a_jZ_{t-k-j}\right)^4\\
  &=\frac{1}{16}E\left(\sum_{j=0}^{\infty}a_jZ_{t-j}+\sum_{j=k}^{\infty}a_{j-k}Z_{t-j}\right)^4\\
  &=\frac{1}{16}E\left(\sum_{j=0}^{k-1}a_jZ_{t-j}+\sum_{j=k}^{\infty}(a_j+a_{j-k})Z_{t-j}\right)^4. \nonumber
\end{split}
\end{equation}
Obviously, $Y_t:=\sum_{j=0}^{k-1}a_jZ_{t-j}+\sum_{j=k}^{\infty}(a_j+a_{j-k})Z_{t-j}$ is also a causal linear process with innovation $\left\{Z_t\right\}$. Since $EZ_t^4<\infty$, it follows that $EY_t^4<\infty$. This proves the desired result. An alternative proof of this Lemma can be given by using a similar argument as in (\ref{e3}).\\
\end{proof}

{\bf Acknowledgement}\\
 The authors would like to thank the editor and the referees for carefully
reading the manuscript and for the suggestions that improved the
presentation.

\newpage
\begin{table}[ht]
\center
\caption{Comparison of MLE, LASSO PCMLE, and SCAD PCMLE for model (\ref{model_1}).}
\label{tab:AR5_1}
\bigskip
\begin{tabular}{crrrrrrrrrr}
\hline
 {\bf Lag} &  {\bf MLE} &  {\bf std} &     {\bf } & {\bf LASSO} &  {\bf std} &     {\bf } & {\bf SCAD} &  {\bf std} &     {\bf } & {\bf True} \\
\hline

         1 &     0.2067 &     0.0179 &            &     0.1947 &     0.0191 &            &     0.2015 &     0.0173 &            &        0.2 \\

         2 &     -0.008 &     0.0179 &            &          0 &          0 &            &          0 &          0 &            &          0 \\

         3 &     0.2191 &     0.0172 &            &      0.207 &     0.0183 &            &     0.2139 &     0.0166 &            &        0.2 \\

         4 &     -0.018 &     0.0182 &            &     -0.001 &      0.001 &            &          0 &          0 &            &          0 \\

         5 &     0.1757 &     0.0181 &            &     0.1637 &     0.0193 &            &     0.1705 &     0.0171 &            &        0.2 \\

           &            &            &            &            &            &            &            &            &            &            \\

{\bf error} &     0.0373 &            &            &     0.0373 &            &            &     0.0326 &            &            &            \\

{\bf $\lambda_N$} &          - &            &            &       0.02 &            &            &       0.08 &            &            &            \\

   {\bf a} &          - &            &            &          - &            &            &        2.1 &            &            &            \\
\hline
\end{tabular}  
\end{table}

\begin{table}[ht]
\center
\caption{Summary of 100 independent simulations for model (\ref{model_1}) with sample size 1000.}
\label{tab:AR5_100}
\bigskip
\begin{tabular}{rrrrrrr}
\hline
           &            & \multicolumn{ 2}{c}{{\bf LASSO}} &     {\bf } & \multicolumn{ 2}{c}{{\bf SCAD}} \\

           &            &     {\bf } &     {\bf } &     {\bf } &     {\bf } &     {\bf } \\

           &            & \multicolumn{ 2}{c}{{\bf probability of 2 zeros}} &     {\bf } & \multicolumn{ 2}{c}{{\bf probability of 2 zeros}} \\

           &            & \multicolumn{ 2}{c}{0.2} &     {\bf } & \multicolumn{ 2}{c}{0.61} \\
\hline
           &            &            &            &            &            &            \\

           &            & {\bf Probability of } & {\bf Average} &     {\bf } & {\bf Probability of } & {\bf Average} \\

 {\bf Lag} & {\bf TRUE} & {\bf  0 estimate} & {\bf bias} &     {\bf } & {\bf  0 estimate} & {\bf bias} \\
\hline
    {\bf } &     {\bf } &     {\bf } &     {\bf } &     {\bf } &     {\bf } &     {\bf } \\

         1 &        0.2 &       0.01 &     0.0298 &            &          0 &     0.0282 \\

         2 &          0 &       0.29 &     0.0197 &            &       0.72 &     0.0063 \\

         3 &        0.2 &          0 &     0.0309 &            &          0 &     0.0297 \\

         4 &          0 &       0.33 &     0.0206 &            &       0.79 &     0.0034 \\

         5 &        0.2 &          0 &     0.0250 &            &          0 &     0.0340 \\
\hline
\end{tabular}  
\end{table}

\newpage
\begin{table}[ht]
\center
\caption{Comparison of MLE, LASSO PMLE, and SCAD PMLE for  model (\ref{model_1}) with student T innovations.}
\label{tab:AR_T2}
\bigskip
\begin{tabular}{rrrrrrrrr}
\hline
  {\bf df} &  {\bf Lag} &  {\bf MLE} &     {\bf } & {\bf LASSO } &     {\bf } & {\bf SCAD } &     {\bf } & {\bf TRUE} \\
\hline
           &            &            &            &            &            &            &            &            \\

         2 &          1 &     0.1995 &            &     0.1978 &            &     0.1435 &            &        0.2 \\

           &          2 &    -0.0029 &            &          0 &            &          0 &            &          0 \\

           &          3 &     0.1801 &            &     0.1799 &            &     0.1642 &            &        0.2 \\

           &          4 &     0.0025 &            &     0.0018 &            &          0 &            &          0 \\

           &          5 &     0.1726 &            &     0.1713 &            &     0.1841 &            &        0.2 \\

           &            &            &            &            &            &            &            &            \\

           & {\bf error} & {\bf 0.0341} &     {\bf } & {\bf 0.0351} &            & {\bf 0.0687} &            &            \\

           &            &            &            &            &            &            &            &            \\

         5 &          1 &     0.1831 &            &     0.1806 &            &     0.1825 &            &        0.2 \\

           &          2 &    -0.0048 &            &          0 &            &          0 &            &          0 \\

           &          3 &     0.2161 &            &     0.2151 &            &     0.2159 &            &        0.2 \\

           &          4 &      0.027 &            &     0.0234 &            &     0.0203 &            &          0 \\

           &          5 &     0.1987 &            &     0.1978 &            &        0.2 &            &        0.2 \\

           &            &            &            &            &            &            &            &            \\

           & {\bf error} & {\bf 0.036} &     {\bf } & {\bf 0.034} &     {\bf } & {\bf 0.0311} &            &            \\
\hline
\end{tabular}  
\end{table}

\newpage
\begin{table}[ht]
\center
\caption{A comparison of forecasting performances of 3 methods: MLE, SCAD PCMLE, and MLE with an optimal order chosen by FPE.}
\label{tab:real2_forc}
\bigskip
\begin{tabular}{rrrrrrrrrr}
\hline
           &            &            &                                                 \multicolumn{ 7}{c}{Forecast Evaluation} \\

           &            &            &            &            &            &            &            &            &            \\

           &            &            &     1-step &     6-step &    12-step &            &     1-step &     6-step &    12-step \\

\multicolumn{ 2}{c}{Model} &            &        MAE &        MAE &        MAE &            &       RMSE &       RMSE &       RMSE \\
\hline
           &            &            &            &            &            &            &            &            &            \\

      p=30 &        MLE &            &     0.0506 &     0.0497 &     0.0478 &            &     0.3921 &      0.385 &     0.3704 \\

           &      SCAD  &            &     0.0414 &     0.0406 &     0.0392 &            &     0.3207 &     0.3145 &     0.3035 \\

           &            &            &            &            &            &            &            &            &            \\

      p=25 &        MLE &            &     0.0454 &     0.0445 &     0.0428 &            &     0.3515 &     0.3448 &     0.3319 \\

           &       SCAD &            &     0.0414 &     0.0406 &     0.0392 &            &      0.321 &     0.3148 &     0.3038 \\

           &            &            &            &            &            &            &            &            &            \\

      p=20 &        MLE &            &     0.0477 &     0.0469 &     0.0453 &            &     0.3694 &      0.363 &     0.3507 \\

           &      SCAD  &            &      0.046 &     0.0452 &     0.0435 &            &     0.3563 &       0.35 &     0.3372 \\

           &            &            &            &            &            &            &            &            &            \\

      p=24 &        FPE &            &     0.0461 &     0.0452 &     0.0435 &            &     0.3571 &     0.3503 &     0.3371 \\
\hline
\end{tabular}  
\end{table}

\newpage
\begin{table}[ht]
\center
\caption{Estimated values of the coefficients for all the combinations of models and methods.}
\label{tab:real2_esti}
\bigskip
\begin{tabular}{rrrrrrrrrrr}
\hline
           &                                                                                    \multicolumn{ 10}{c}{Estimated Coefficients} \\

           &            &            &            &            &            &            &            &            &            &            \\

           & \multicolumn{ 2}{c}{p=30} &            & \multicolumn{ 2}{c}{p=25} &            & \multicolumn{ 2}{c}{p=20} &            &       p=24 \\

       Lag &        MLE &       SCAD &            &        MLE &       SCAD &            &        MLE &       SCAD &            &        FPE \\
\hline
    {\bf } &     {\bf } &     {\bf } &            &     {\bf } &     {\bf } &            &     {\bf } &     {\bf } &            &     {\bf } \\

         1 &    -0.0486 &          0 &            &    -0.0458 &          0 &            &    -0.0549 &          0 &            &    -0.0431 \\

         2 &     0.0955 &     0.1031 &            &     0.0928 &     0.1054 &            &     0.0867 &     0.0966 &            &     0.0923 \\

         3 &      0.083 &     0.0819 &            &     0.0864 &      0.084 &            &     0.0948 &     0.0867 &            &     0.0864 \\

         4 &     0.0519 &          0 &            &     0.0465 &          0 &            &     0.0461 &          0 &            &     0.0475 \\

         5 &    -0.0212 &          0 &            &    -0.0236 &          0 &            &    -0.0225 &          0 &            &    -0.0239 \\

         6 &     0.0594 &          0 &            &     0.0636 &          0 &            &     0.0524 &     0.0639 &            &     0.0627 \\

         7 &    -0.0146 &          0 &            &    -0.0134 &          0 &            &    -0.0178 &          0 &            &     -0.015 \\

         8 &    -0.0021 &          0 &            &     0.0028 &          0 &            &     0.0081 &          0 &            &      0.003 \\

         9 &     0.0903 &     0.0969 &            &     0.0935 &     0.0964 &            &      0.092 &     0.0935 &            &     0.0944 \\

        10 &     0.0499 &          0 &            &     0.0535 &          0 &            &     0.0546 &      0.058 &            &     0.0535 \\

        11 &    -0.0031 &          0 &            &    -0.0009 &          0 &            &     0.0004 &          0 &            &    -0.0012 \\

        12 &     0.0432 &          0 &            &     0.0409 &          0 &            &     0.0318 &          0 &            &     0.0409 \\

        13 &     0.0087 &          0 &            &     0.0065 &          0 &            &      0.009 &          0 &            &     0.0056 \\

        14 &    -0.0124 &          0 &            &    -0.0115 &          0 &            &    -0.0109 &          0 &            &    -0.0116 \\

        15 &     0.0039 &          0 &            &     0.0034 &          0 &            &    -0.0114 &          0 &            &     0.0028 \\

        16 &    -0.0446 &          0 &            &    -0.0405 &          0 &            &    -0.0406 &          0 &            &     -0.042 \\

        17 &    -0.0093 &          0 &            &    -0.0025 &          0 &            &     0.0002 &          0 &            &    -0.0021 \\

        18 &     0.0827 &     0.0936 &            &     0.0865 &     0.0931 &            &     0.0736 &     0.0715 &            &     0.0867 \\

        19 &     0.0359 &          0 &            &      0.042 &          0 &            &      0.046 &          0 &            &     0.0409 \\

        20 &     0.0369 &          0 &            &     0.0429 &          0 &            &     0.0461 &          0 &            &     0.0435 \\

        21 &    -0.0416 &          0 &            &    -0.0425 &          0 &            &            &            &            &    -0.0433 \\

        22 &     0.0175 &          0 &            &      0.023 &          0 &            &            &            &            &     0.0217 \\

        23 &     0.0708 &     0.0768 &            &     0.0751 &     0.0768 &            &            &            &            &     0.0736 \\

        24 &    -0.1508 &    -0.1373 &            &    -0.1439 &    -0.1363 &            &            &            &            &    -0.1435 \\

        25 &    -0.0219 &          0 &            &    -0.0169 &          0 &            &            &            &            &            \\

        26 &     0.0368 &          0 &            &            &            &            &            &            &            &            \\

        27 &     0.0063 &          0 &            &            &            &            &            &            &            &            \\

        28 &     0.0519 &          0 &            &            &            &            &            &            &            &            \\

        29 &     0.0313 &          0 &            &            &            &            &            &            &            &            \\

        30 &     0.0069 &          0 &            &            &            &            &            &            &            &            \\
\hline
\end{tabular}  
\end{table}

%
%

\end{document}